\documentclass{lmcs} 
\pdfoutput=1

\usepackage{lastpage}
\lmcsdoi{17}{1}{9}
\lmcsheading{}{\pageref{LastPage}}{}{}%
{Dec.~20,~2019}{Feb.~02,~2021}{}


\usepackage{booktabs}

\usepackage{xurl}
\usepackage[utf8]{inputenc}

\usepackage{amsfonts,amssymb,latexsym,stmaryrd,mathtools}
\usepackage{xifthen}
\usepackage{multirow}
\usepackage[linesnumbered,ruled]{algorithm2e}

\usepackage{microtype}

\usepackage{tikz}
\usetikzlibrary{calc,positioning}
\usetikzlibrary{decorations.markings,arrows,decorations.pathreplacing}
\tikzstyle{vertex} = [draw, circle, minimum size=0.6cm]
\tikzstyle{bullet} = [fill, color=black, inner sep=0pt, minimum size=0.1cm, circle]
\tikzstyle{vertexbase} = [minimum size=6.5mm, rounded corners=1mm]
\tikzstyle{min} = [draw, fill=red!20, rectangle, vertexbase]
\tikzstyle{max} = [draw, fill=blue!20, rectangle, vertexbase]
\tikzstyle{rnd} = [draw, rectangle, densely dotted, thick, vertexbase]
\tikzstyle{onesink} = [draw, fill=black!20, rectangle, vertexbase]
\tikzstyle{zerosink} = [draw, rectangle, vertexbase]

\theoremstyle{plain} 

\def\eg{\textit{e.g.}}
\def\cf{\textit{cf.}}

\providecommand*{\ifempty}[3]{\ifthenelse{\isempty{#1}}{#2}{#3}}
\newcommand{\parensmathoper}[2]{\ensuremath{#1\ifempty{#2}{}{(#2)}}}

\newcommand{\Distr}[1]{\mathcal{D}(#1)} 
\newcommand{\argmin}{\mathop{\text{argmin}}}
\newcommand{\C}{\mathcal{C}}
\newcommand{\bdist}[1][\lambda]{\mathbf{d}_{#1}}
\renewcommand{\H}[1]{\parensmathoper{\mathcal{H}}{#1}}
\newcommand{\K}[1]{\parensmathoper{\mathcal{K}}{#1}}

\newcommand{\supp}{\mathop{supp}}
\newcommand{\dirac}{\mathbf{1}}

\newcommand{\tv}{\parensmathoper{\mathop{TV}}} 
\renewcommand{\Pr}[2][]{\parensmathoper{\mathrm{Pr}_{#2}^{#1}}} 

\newcommand{\change}[1]{#1}
\newcommand{\franck}[1]{#1}

\newcommand{\lopen}[1]{(#1]} 

\begin{document}

\title[Computing Probabilistic Bisimilarity Distances for PAs]{Computing Probabilistic Bisimilarity Distances \texorpdfstring{\\}{} for Probabilistic Automata}
\titlecomment{{\lsuper*} This paper is an extended version of an earlier conference paper~\cite{BacciBLMTB19} presented at CONCUR 2019.}

\author[G.~Bacci]{Giorgio Bacci\rsuper{a}}	
\author[G.~Bacci]{Giovanni Bacci\rsuper{a}}	
\author[K.~G.~Larsen]{Kim G.~Larsen\rsuper{a}}	
\author[R.~Mardare]{\texorpdfstring{\\}{}Radu Mardare\rsuper{b}}	
\author[Q.~Tang]{Qiyi Tang\rsuper{c}}	
\author[F.~van Breugel]{Franck van Breugel\rsuper{d}\texorpdfstring{\vspace{-25mm}}{}}	

\address{\lsuper{a}Dept.~of Computer Science, Aalborg University, Denmark}	
\email{\{grbacci,giovbacci,kgl\}@cs.aau.dk}  

\address{\lsuper{b}Dept.~of Computer and Information Sciences, University of Strathclyde, Glasgow, UK}	
\email{r.mardare@strath.ac.uk}  

\address{\lsuper{c}Dept.~of Computer Science, Oxford University, UK}	
\email{qiyi.tang@cs.ox.ac.uk}  

\address{\lsuper{d}DisCoVeri Group, Dept.~of Electrical Engineering and Computer Science, York University, Canada}	
\email{franck@eecs.yorku.ca}  

\thanks{Giovanni Bacci and Kim G.~Larsen are supported by the ERC-Project LASSO}	

\thanks{Franck van Breugel is supported by Natural Sciences and Engineering Research Council of Canada}




\begin{abstract}
  \noindent The probabilistic bisimilarity distance of Deng et al.~has been proposed as a robust quantitative generalization of Segala and Lynch's probabilistic bisimilarity for probabilistic automata.
In this paper, we present a characterization of the bisimilarity distance as the solution of a simple stochastic game. The characterization gives us an algorithm to compute the distances by applying Condon's simple policy iteration on these games. The correctness of Condon's approach, however, relies on the assumption that the games are \emph{stopping}. Our games may be non-stopping in general, yet we are able to prove termination for this extended class of games. Already other algorithms have been proposed in the literature to compute these distances, with complexity in $\textbf{UP} \cap \textbf{coUP}$ and \textbf{PPAD}. Despite the theoretical relevance, these algorithms are inefficient in practice. To the best of our knowledge, our algorithm is the first practical solution.

The characterization of the probabilistic bisimilarity distance mentioned above crucially uses a dual presentation of the Hausdorff distance due to M\'emoli. As an additional contribution, in this paper we
show that M\'emoli's result can be used also to prove that the bisimilarity distance bounds the difference in the
maximal (or minimal) probability of two states to satisfying arbitrary $\omega$-regular properties, expressed, eg., as LTL formulas.
\end{abstract}

\maketitle

\section{Introduction}%
\label{sec:intro}

In~\cite{GiacaloneJS90}, Giacalone et al.\ observed that for reasoning about the behaviour of probabilistic systems, rather than equivalences, a notion of distance is more reasonable in practice since it permits to
capture the degree of difference between two states. This observation motivated the study of
behavioural pseudometrics, that generalize behavioural equivalences in the sense that, when the
distance is zero then the two states are behaviourally equivalent.

The systems we consider in this paper are labelled \emph{probabilistic automata}.
This model was introduced by Segala~\cite{Segala95} to capture both nondeterminism (hence, concurrency)
and probabilistic behaviours.
The labels on states are used to express that certain properties of interest hold in particular states.

\change{In Figure~\ref{fig:gamblers} we consider an example of a probabilistic automaton describing two gamblers,
$f$ and $b$, deciding on which team to bet in a football match.}
\begin{figure}[b]
\centering
\begin{tikzpicture}

\node[vertex] (a) at (0,1) {$f$};
\node[vertex, fill=blue!20] (d) at (2,0) {$t$};
\node[vertex] (b) at (4,1) {$b$};
\node[vertex,fill=red!50] (c) at (2,2) {$h$};

\node[bullet] (e) at (1,1) {};
\node[bullet] (f) at (3,1) {};

\draw[-] (a)--(e);
\draw[->] (e)--(c) node [pos=0.2, above] {$\frac{1}{2}$};
\draw[->] (e)--(d) node [pos=0.2, below] {$\frac{1}{2}$};

\draw[-] (b)--(f);
\draw[->] (f)--(c) node [pos=0.2, xshift=1ex, above] {$\frac{51}{100}$};
\draw[->] (f)--(d) node [pos=0.2, xshift=1ex, below] {$\frac{49}{100}$};

\path[->] (c) edge [loop above, min distance=10mm] node[bullet, above=-0.6mm] {}
 node[below right, xshift=1mm] {$1$} (c);
\path[->] (d) edge [loop below, min distance=10mm] node[bullet, above=-0.6mm] {}
node [above left, xshift=-1mm] {$1$} (d);

\draw[->] (a) to[bend left] node[pos=0.7, above] {$1$} node[pos=0.5,bullet] {} (c);
\draw[->] (b) to[bend right] node[pos=0.7, above] {$1$} node[pos=0.5,bullet] {} (c);
\draw[->] (a) to[bend right] node[pos=0.7, below] {$1$} node[pos=0.5,bullet] {} (d);
\draw[->] (b) to[bend left] node[pos=0.7, below] {$1$} node[pos=0.5,bullet] {} (d);

\end{tikzpicture}

\caption{A probabilistic automaton describing two gablers.}%
\label{fig:gamblers}
\end{figure}
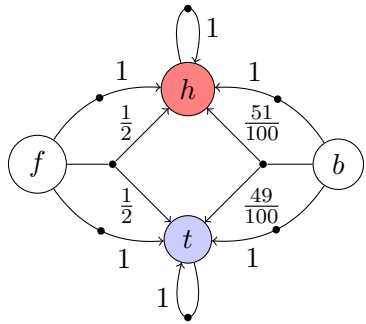
Typically the two gamblers know on which team to bet, but occasionally they prefer to toss a coin to make a decision.
This is represented by the three probabilistic transitions in the state $f$. The first two take $f$ to state $h$ (head) or $t$ (tail) with probability one, the last takes $f$ to states $h$ and $t$ with probability $\frac{1}{2}$ each.
The difference between $f$ and $b$ is that the former uses a fair coin while the latter uses a biased coin landing on heads with slightly higher probability.
Once the decision is taken, it is not changed anymore. This is seen on states $h$ and $t$ which have a single probabilistic transition taking the state to itself with probability one. The states $h$ and $t$ have distinct labels, here represented by colours.

A behavioural pseudometric for probabilistic automata capturing this difference is the
\emph{probabilistic bisimilarity distance} by Deng et al.~\cite{DengCPP06}, introduced as a robust
generalization of Segala and Lynch's probabilistic bisimilarity~\cite{SegalaL94}.
The key ingredients of this pseudometric are the Hausdorff metric~\cite{Hausdorff14} and the Kantorovich metric~\cite{Kantorovich42}, respectively used to capture nondeterministic and probabilistic behaviour.
In the example above, the behaviours of the states $h$ and $t$ are very different since their labels are different.
As a result, their probabilistic bisimilarity distance is one. On the other hand, the behaviours of the states
$f$ and $b$ are very similar, which is reflected by the fact that their probabilistic bisimilarity distance is
$\frac{1}{100}$.

The first attempt to compute the above distance is due
to Chen et al.~\cite{ChenHL07}, who proposed a doubly exponential-time procedure to
approximate the distances up to any degree of accuracy. The complexity was later improved to
$\textbf{PSPACE}$ by Chattarjee et al.~\cite{ChatterjeeAMR08,ChatterjeeAMR10}. Their solutions exploit
the decision procedure for the existential fragment of the first-order theory of \franck{the} reals. It is worth noting that~\cite{ChenHL07,ChatterjeeAMR08} consider the pseudometric that does not discount the future (a.k.a.\ \emph{undiscounted} distance) which
entails additional algorithmic challenges. Later, Fu~\cite{Fu12} showed that the distances have rational values and that computing the discounted distance can be done in polynomial time by using a value-iteration procedure in combination with the continued fraction algorithm~\cite[Section 6]{Schrijver}.
As for the undiscounted distance, he showed that the threshold problem, i.e., deciding whether the distance is smaller than a given rational, is in $\textbf{NP} \cap \textbf{coNP}$. The same proof can be adapted to show that the decision problem is in $\textbf{UP} \cap \textbf{coUP}$~\cite{Fu13}, where $\textbf{UP}$ is the subclass of $\textbf{NP}$-problems with a unique accepting computation.
Van Breugel and Worrell~\cite{BreugelW14} have later shown that the problem is in $\textbf{PPAD}$, which is short for polynomial parity argument on directed graphs.
Notably, their proof exploits a characterization of the distance as a simple stochastic game.
The above algorithms were presented with the purpose of understanding the complexity of computing bisimilarity distances and, to the best of our knowledge, they have never been implemented. Their implementation would involve either an enumeration of possibly exponentially many fixed-points~\cite{Fu13}, or the use of SMT solvers over the existential fragment of the first order theory of the reals~\cite{ChatterjeeAMR08,ChatterjeeAMR10}. An earlier attempt of approximating the bisimilarity distance for the more specific case of labelled Markov chains by expressing the problem in the existential fragment of the first order theory of the reals was proposed in~\cite{BreugelSW08}. Its latest implementation using CVC4~\cite{BCD+11} is able to handle chains with 82 states in approximately 66 hours\footnote{The code is available at \url{bitbucket.org/discoveri/first-order}.}.
In this paper, we propose an alternative approach that is inspired by the successful implementations of similar pseudometrics on labelled Markov chains~\cite{BacciLM:tacas13,TangB16,TangB18}.

Our solution is based on a novel characterization of the probabilistic bisimilarity distance as the solution of a
simple stochastic game. Stochastic games were introduced by Shapley~\cite{Shapley53}. A simplified version of these games, called \emph{simple stochastic games}, were studied by Condon~\cite{Condon92}. Several algorithms have been proposed to compute the value function of a simple stochastic game, many using policy iteration.
Condon~\cite{Condon90} proposed an algorithm, known as \emph{simple policy iteration}, that switches
only one non-optimal choice per iteration.
The correctness of Condon's algorithm, however, relies on the assumption that the game is \emph{stopping}.

It turns out that the simple stochastic games characterizing the probabilistic bisimilarity
distances are stopping only when the distances discount the future.
In the case the distance is non-discounting, the corresponding games may not be stopping.
To recover correctness of the policy iteration procedure we adapt Condon's simple policy iteration
algorithm by adding a non-local update of the strategy of the min player and an extra termination condition  based on a notion of ``self-closed'' relation due to Fu~\cite{Fu12}.
The practical efficiency of our algorithm has been evaluated on a significant set of randomly generated
probabilistic automata. The results show that our algorithm performs better than the corresponding
iterative algorithms proposed for the discounted distances in~\cite{Fu12}, even though the theoretical
complexity of our proposal is exponential in the worst case (\cf~\cite{TangB16}) whereas Fu's is polynomial.
The implementation of the algorithms exploits a coupling structure
characterization of the distance that allows us to skip the construction of the simple stochastic
game which may result in an exponential blow up of the memory required for storing the game.

The two characterizations of probabilistic bisimilarity distances proposed in this paper (either via simple stochastic games or coupling structures) crucially use a dual presentation of
the Hausdorff distance due to M\'emoli~\cite{Memoli11}. Still using M\'emoli's result, as an additional contribution to this paper we show that the (undiscounted) bisimilarity distance can be used to bound the difference of the maximal (or minimal) probability of two states satisfying arbitrary $\omega$-regular specifications, expressed, \franck{\eg}, as LTL formulas.
Notably, this result allows us to relate the probabilistic bisimilarity pseudometric of Deng et al.\ to probabilistic model checking of probabilistic automata against linear-time specifications.

\subsubsection*{Synopsis} Section~\ref{sec:prelim} introduces the notation and some preliminary results used in the paper. In Section~\ref{sec:pautomata} we recall the definition of probabilistic bisimilarity distances of Deng et al.\ for probabilistic automata;  then, in Section~\ref{sec:ssg} we propose a characterisation of the probabilistic bisimilarity distances as the values of a simple stochastic game constructed from the automaton, here called probabilistic bisimilarity game. Towards an algorithmic solution for computing bisimilarity distances, in Section~\ref{sec:couplingdistance} we provide an alternative characterisation of the distances in terms of \emph{coupling structures}. Section~\ref{sec:algorithm} describes a procedure for computing the bisimilarity distances based on Condon's simple policy iteration algorithm. In Section~\ref{sec:linearProperties} we discuss the relation between the notion of bisimilarity distance and probabilistic model checking of $\omega$-regular linear-time specifications against probabilistic automata. Finally, Section~\ref{sec:conclusion} concludes with some remarks and future work directions.

\section{Preliminaries and Notation}%
\label{sec:prelim}

The set of functions $f$ from $X$ to $Y$ is denoted by $Y^X$. We denote by $f[x/y] \in Y^X$
the \emph{update of $f$} at $x \in X$ with $y \in Y$, defined by $f[x/y](x') = y$ if $x' = x$,
otherwise $f[x/y](x') = f(x')$.

A (1-bounded) \emph{pseudometric} on a set $X$ is a function $d \colon X \times X \to [0,1]$ such that,
$d(x,x) = 0$, $d(x,y) = d(y,x)$, and $d(x,y) \leq d(x,z) + d(z,y)$ for all $x, y, z \in X$.

\subparagraph{Kantorovich lifting.}
A (discrete) probability distribution
on $X$ is a function $\mu \colon X \to [0,1]$ such that $\sum_{x \in X} \mu(x) = 1$, and its support is $\supp(\mu) = \{ x \in X \mid \mu(x) > 0 \}$.  We denote by
$\Distr{X}$ the set of probability distributions on $X$.
A pseudometric $d$ on $X$ can be 
lifted to a pseudometric on probability distributions in $\Distr{X}$ by means of the Kantorovich lifting~\cite{Villani}.

The \emph{Kantorovich lifting} of $d \in {[0,1]}^{X \times X}$ on distributions $\mu, \nu \in \Distr{X}$ is defined by
\change{\begin{equation}
\K{d}(\mu,\nu) =
\sup \left\{ \sum_{x \in X} f(x) \cdot \franck{(\mu(x) - \nu(x))}  \mid f \in L_d  \right\} \,,
\tag{\sc Kantorovich lifting}
\end{equation}
where $L_d$ denotes the set of non-expansive $[0,1]$-valued functions over $X$, i.e., functions
$f \colon X \to [0,1]$ such that, for all $x,y \in X$, $|f(x) - f(y)| \leq d(x,y)$.

The Kantorovich distance has the following well know dual formulation
\begin{equation*}
\K{d}(\mu,\nu) = \min \left \{ \sum_{x,y \in X} d(x,y) \cdot \omega(x,y) \mid  \omega \in \Omega(\mu,\nu)  \right \} \,,
\end{equation*}}
where $\Omega(\mu,\nu)$ denotes the set of \emph{measure-couplings} for the pair $(\mu, \nu)$, i.e., distributions $\omega \in \Distr{X \times X}$ such that, for all $x \in X$, $\sum_{y \in X} \omega(x, y) = \mu(x)$ and $\sum_{y \in X} \omega(y, x) = \nu(x)$.
It is a well known fact that this dual characterisation can be equivalently stated
by ranging $\omega$ over the set of vertices $V(\Omega(\mu, \nu))$ of the polytope $\Omega(\mu, \nu)$, \change{as a concave function on a convex polytope attains its minimum at a vertex of the polytope (see~\cite[page~260]{KleeW67}).} Furthermore, if the set $X$ is finite, the set $V(\Omega(\mu, \nu))$ is finite too\change{~\cite[page~259]{KleeW67}}.

\subparagraph{Hausdorff lifting.}
A pseudometric $d$ on $X$ can be lifted to nonempty subsets of $X$ by means
of the Hausdorff lifting. The \emph{Hausdorff lifting} of $d \in {[0,1]}^{X \times X}$ on nonempty subsets $A, B \subseteq X$
is defined by
\begin{equation}
\H{d}(A,B) = \max \left\{  \adjustlimits \sup_{a \in A} \inf_{b \in B} d(a,b), \,
\adjustlimits \sup_{b \in B} \inf_{a \in A} d(a,b) \right\} \, .
\tag{\sc Hausdorff lifting}
\end{equation}

Following M\'emoli~\cite[Lemma~3.1]{Memoli11}, the Hausdorff lifting has a dual characterization
in terms of \emph{set-couplings}%
\footnote{M\'emoli uses the terminology ``correspondence.'' To avoid confusion, we adopted the same terminology used in~\cite[Section~10.6]{PeyreC19}.}. Given $A,B \subseteq X$, a set-coupling for $(A,B)$ is a relation
$R \subseteq X \times X$ with left and right projections respectively equal to $A$ and $B$, i.e., $\{ a \mid \exists b \in X.\, a \mathrel{R} b \} = A $ and $\{ b \mid \exists a \in X. \, a \mathrel{R} b \} = B$. We write $\mathcal{R}(A,B)$ for the set of the set-couplings for $(A,B)$.
\begin{thmC}[\cite{Memoli11}]%
\label{th:hausdorff}
$\H{d}(A,B) = \inf \{ \sup_{(a,b) \in R} d(a,b) \mid R \in \mathcal{R}(A,B) \}$.
\end{thmC}
Clearly, for finite $A,B$, $\inf$ and $\sup$ in Theorem~\ref{th:hausdorff} can be replaced by $\min$ and $\max$, respectively.

\section{Probabilistic Automata and Probabilistic Bisimilarity Distance}%
\label{sec:pautomata}

\change{In this section we recall some definitions and results from the literature. In particular, we introduce the models of interest, \emph{probabilistic automata}, its best known behavioural equivalence, Segala and Lynch's \emph{probabilistic bisimilarity}~\cite{SegalaL94}, and its quantitative generalization due to
Deng et al.~\cite{DengCPP06}.

A probabilistic automaton is a model of computation that combines nondeterministic and probabilistic
behaviours. Similarly to a standard nondeterministic automaton, states are labelled to express that certain properties of interest
hold in that state. A probabilistic automaton in a current state $s \in S$ can nondeterministically proceed to a next probabilistic state $\mu \in \Distr{S}$, representing the probability distribution with which the automaton will move to the next state. This can be formalised as follows:}
\begin{defi}
A \emph{probabilistic automaton} (PA) is a tuple $\mathcal{A} = (S, L, \to, \ell)$ consisting of a nonempty finite set $S$ of states, a finite set of labels $L$, a finite total transition relation ${\to} \subseteq S \times \Distr{S}$, and a labelling function $\ell \colon S \to L$.
\end{defi}
For simplicity we assume the transition relation $\to$ to be total, that is, for all $s \in S$, there exists
a $\mu \in \Distr{S}$ such that $(s,\mu) \in {\to}$.
For the remainder of this paper we fix a probabilistic automaton $\mathcal{A} = (S, L, \to, \ell)$. We write $s \to \mu$ to denote $(s,\mu) \in {\to}$ and use $\delta(s)$ to denote the set $\{ \mu \mid s \to \mu\}$ of successor distributions of $s$.

Next we recall the notion of probabilistic bisimilarity due to Segala and Lynch~\cite{SegalaL94} for probabilistic automata. Their definition exploits the notion of lifting of a relation $R \subseteq S \times S$ on states to a relation   $\tilde{R} \subseteq \Distr{S} \times \Distr{S}$ on probability distributions on states,
originally introduced by Jonsson and Larsen~\cite{JonssonL91}, and defined by
$\mu \mathrel{\tilde{R}} \nu$ if there exists a measure-coupling $\omega \in \Omega(\mu,\nu)$ such that
$\supp(\omega) \subseteq R$.
\begin{defi}%
\label{def:bisimulation}
A relation $R \subseteq S \times S$ is a \emph{probabilistic bisimulation} if whenever $s \mathrel{R} t$,
\begin{itemize}
\item $\ell(s) = \ell(t)$,
\item if $s \to \mu$ then there exists $t \to \nu$ such that $\mu \mathrel{\tilde{R}} \nu$, and
\item if $t \to \nu$ then there exists $s \to \mu$ such that $\mu \mathrel{\tilde{R}} \nu$.
\end{itemize}
Two states $s,t \in S$ are \emph{probabilistic bisimilar}, written $s \sim t$, if they are related by some
probabilistic bisimulation.
\end{defi}
Intuitively, two states are probabilistic bisimilar if they have the same label and each transition of the one state
to a distribution $\mu$ can be matched by a transition of the other state to a distribution $\nu$ assigning the same probability to states that behave the same, and vice versa.
Probabilistic bisimilarity is an equivalence relation and the largest probabilistic bisimulation.

Deng et al.~\cite{DengCPP06} proposed a family of 1-bounded pseudometrics $\bdist$, parametric on a \emph{discount factor} $\lambda \in \lopen{0,1}$, called \emph{probabilistic bisimilarity pseudometrics}. The pseudometrics $\bdist$ are defined as the least fixed-point of the functions $\Delta_\lambda \colon {[0,1]}^{S \times S} \to {[0,1]}^{S \times S}$
\begin{equation*}
\Delta_\lambda(d)(s,t) =
\begin{cases}
1 & \text{if $\ell(s) \neq \ell(t)$} \\
\lambda \cdot \H{\K{d}}(\delta(s),\delta(t)) & \text{otherwise} \,.
\end{cases}
\end{equation*}
The well-definition of $\bdist$ follows by Knaster-Tarski's fixed point theorem, given the fact that
$\Delta_\lambda$ is a monotone function on the complete partial order of $[0,1]$-valued functions
on $S \times S$ ordered point-wise by $d \sqsubseteq d'$ iff for all $s,t \in S$, $d(s,t) \leq d'(s,t)$.

The fact that probabilistic bisimilarity distances provide a quantitative generalization of
bisimilarity is captured by the following theorem due to Deng et al.~\cite[Corollary~2.14]{DengCPP06}.
\begin{thm}%
\label{th:behaviouraldist}
For all $\lambda \in \lopen{0,1}$, $\bdist(s,t) = 0$ if and only if $s \sim t$.
\end{thm}


\section{Probabilistic Bisimilarity Distance as a Simple Stochastic Game}%
\label{sec:ssg}

%
A \emph{simple stochastic game} (SSG) consists of a finite directed graph whose vertices are partitioned into sets of \emph{0-sinks}, \emph{1-sinks}, \emph{max vertices}, \emph{min vertices}, and \emph{random vertices}.  The game is played by two players, the \emph{max player} and the \emph{min player}, with a single token.  At each step of the game, the token is moved from a vertex to one of its successors.  At a min vertex the min player chooses the successor, at a max vertex the max player chooses the successor, and at a random vertex the successor is chosen randomly according to a prescribed probability distribution.  The max player wins a play of the game if the token reaches a 1-sink and the min player wins if the play reaches a 0-sink or continues forever without reaching a sink.  Since the game is stochastic, the max player tries to maximize the probability of reaching a 1-sink whereas the min player tries to minimize that probability.

\begin{defi}
A \emph{simple stochastic game} is a tuple $(V, E, P)$ consisting
of
\begin{itemize}
\item
a finite directed graph $(V, E)$ such that
\begin{itemize}
\item
$V$ is partitioned into the sets: $V_0$ of \emph{0-sinks}, $V_1$ of \emph{1-sinks},
$V_{\max}$ of \emph{max vertices}, $V_{\min}$ of \emph{min vertices}, and $V_{\mathrm{rnd}}$ of \emph{random vertices};
\item
the vertices in $V_0$ and $V_1$ have outdegree zero and all other
vertices have outdegree at least one, and
\end{itemize}
\item
a function $P \colon V_{\mathrm{rnd}} \to \Distr{V}$ such that for all $v \in V_{\mathrm{rnd}}$ and $w \in V$, $P(v)(w) > 0$ iff $(v, w) \in E$.
\end{itemize}
\end{defi}

\noindent
\change{The above definition is slightly more general than the one given by Condon in~\cite[Section~2]{Condon92}. Note that the outdegree of min, max and random vertices is at least one (instead of exactly two),
there may be multiple 0-sinks and 1-sinks (rather than exactly one). However, a simple stochastic game
as defined above can be transformed in polynomial-time into a simple stochastic game as defined in~\cite{Condon92}, as shown by Zwick and Paterson~\cite{ZwickP96}.}

A \emph{strategy}, also known as \emph{policy}, for the min player is a function $\sigma_{\min} \colon V_{\min} \to V$ that assigns the target of an outgoing edge to each min vertex, that is, for all $v \in V_{\min}$,
$(v, \sigma_{\min}(v)) \in E$.  Likewise, a strategy for the max player is a function
$\sigma_{\max} \colon V_{\max} \to V$ that assigns the target of an outgoing edge to each max vertex.
These strategies are known as \emph{pure stationary} strategies.  We can restrict ourselves to these strategies \change{since the optimal among all strategies for both players are} of this type (see, for example,~\cite{LL69}).

Such strategies determine a sub-game in which each max vertex and each min vertex has outdegree one (see~\cite[Section~2]{Condon92} for details).  Such a game can naturally be viewed as a Markov chain.  We write $\phi_{\sigma_{\min},\sigma_{\max}} \colon V \to [0,1]$ for the function that gives the probability
of a vertex in this Markov chain to reach a 1-sink.

The \emph{value function} $\phi \colon V \to [0,1]$ of a SSG is defined as $\min_{\sigma_{\min}} \max_{\sigma_{\max}} \phi_{\sigma_{\min},\sigma_{\max}}$.  It is folklore that the value function of a simple stochastic game can be characterised as the least fixed point of the following function (see, for example,~\cite[Section~2.2 and 2.3]{J05}).

\begin{defi}%
\label{definition:Phi}
The function $\Phi \colon {[0,1]}^{V} \to {[0,1]}^{V}$ is defined by
\begin{equation*}
\Phi(f)(v) = \begin{cases}
\; 0 & \text{if $v \in V_0$}\\
\; 1 & \text{if $v \in V_1$}\\
\; \max_{(v, w) \in E} f(w) & \text{if $v \in V_{\max}$}\\
\; \min_{(v, w) \in E} f(w) & \text{if $v \in V_{\min}$}\\
\; \sum_{(v, w) \in E} P(v)(w)\, f(w) & \text{if $v \in V_{\mathrm{rnd}}$}
\end{cases}
\end{equation*}
\end{defi}

\change{The $[0,1]$-valued functions on $V$ can be ordered point-wise by $f \sqsubseteq g$ iff for all $v \in V$,
$f(v) \leq g(v)$. This partial order is complete in ${[0,1]}^V$, with meet and join respectively given by
the point-wise infimum and supremum.

Then, the existence the least fixed point of $\Phi$ is ensured by Knaster-Tarski's fixed point theorem
and the following result.
\begin{prop}%
\label{prop:monotonePhi}
The function $\Phi$ is monotone.
\end{prop}}
\begin{proof}
Let $f,g \in {[0,1]}^{V}$ and $f \sqsubseteq g$.
Let $v \in V$. It suffices to show that $\Phi(f)(v) \leq \Phi(g)(v)$.  We distinguish the following cases.
\begin{itemize}
\item
If $v \in V_0$ then $\Phi(f)(v) = 0 = \Phi(g)(v)$.
\item
If $v \in V_1$ then $\Phi(f)(v) = 1 = \Phi(g)(v)$.
\item
If $v \in V_{\max}$ then $\Phi(f)(v) = \max_{(v, w) \in E} f(w) \leq \max_{(v, w) \in E} g(w) = \Phi(g)(v)$.
\item
If $v \in V_{\min}$ then $\Phi(f)(v) = \min_{(v, w) \in E} f(w) \leq \min_{(v, w) \in E} g(w) = \Phi(g)(v)$.
\item
If $v \in V_{\mathrm{rnd}}$ then
    \[
    \Phi(f)(v) = \sum_{(v, w) \in E} P(v)(w) \, f(w) \leq \sum_{(v, w) \in E} P(v)(w) \, g(w) = \Phi(g)(v).
    \qedhere
    \]
\end{itemize}
\end{proof}

\noindent
\change{The set ${[0,1]}^V$ can be turned into a Banach space by means of the supremum norm $\| f \| = \max_{v \in V} f(v)$. Recall that a function $F \colon {[0,1]}^V \to {[0,1]}^V$ is non-expansive if for all $f, g \in {[0,1]}^V$, $\| f - g \| \geq \| F(f) - F(g) \|$.

\begin{prop}%
\label{prop:nonexpPhi}
The function $\Phi$ is nonexpansive.
\end{prop}}
\begin{proof}
Let $f,g \in {[0,1]}^{V}$. Let $v \in V$.  It suffices to show that $| \Phi(f)(v) - \Phi(g)(v) | \leq \| f - g \|$.  We distinguish the following cases.
\begin{itemize}
\item
If $v \in V_0$ then $|\Phi(f)(v) - \Phi(g)(v)| = |0-0| = 0 \leq \| f - g \|$.
\item
If $v \in V_1$ then $|\Phi(f)(v) - \Phi(g)(v)| = |1-1| = 0 \leq \| f - g \|$.
\item
Let $v \in V_{\max}$.  Without loss of generality, assume that $\max_{(v, w) \in E} f(w) \geq \max_{(v, w) \in E} g(w)$.  Then
\begin{align*}
|\Phi(f)(v) - \Phi(g)(v)|
& =  \left |\max_{(v, w) \in E} f(w) - \max_{(v, w) \in E} g(w) \right |\\
& = \max_{(v, w) \in E} f(w) - \max_{(v, w) \in E} g(w)\\
& = f(x) - \max_{(v, w) \in E} g(w)\\
& \leq f(x) - g(x)\\
& \leq \| f - g \| \,,
\end{align*}
\change{where $x \in V$ realises the maximum of $\{ f(w) \mid (v,w) \in E \}$.}
\item
Let $v \in V_{\min}$.  Without loss of generality, assume that $\min_{(v, w) \in E} f(w) \geq \min_{(v, w) \in E} g(w)$.  Then
\begin{align*}
|\Phi(f)(v) - \Phi(g)(v)|
& =  \left |\min_{(v, w) \in E} f(w) - \min_{(v, w) \in E} g(w) \right |\\
& = \min_{(v, w) \in E} f(w) - \min_{(v, w) \in E} g(w)\\
& = \min_{(v, w) \in E} f(w) - g(x)\\
& \leq f(x) - g(x)\\
& \leq \| f - g \| \,,
\end{align*}
\change{where $x \in V$ realises the minimum of $\{ g(w) \mid (v,w) \in E \}$.}
\item
If $v \in V_{\mathrm{rnd}}$ then
\begin{align*}
|\Phi(f)(v) - \Phi(g)(v)|
& =  \left |\sum_{(v, w) \in E} P(v)(w) \, f(w) - \sum_{(v, w) \in E} P(v)(w) \, g(w) \right |\\
& =  \left |\sum_{(v, w) \in E} P(v)(w) \, (f(w) - g(w)) \right |\\
& \leq \sum_{(v, w) \in E} P(v)(w) \, \| f - g \|\\
& \leq \| f - g \| \,.
\qedhere
\end{align*}
\end{itemize}
\end{proof}

\subsection{A Probabilistic Bisimilarity Game.}

Fix a probabilistic automaton $\mathcal{A}$ and $\lambda \in \lopen{0,1}$. We will characterise the probabilistic bisimilarity distances as values of a simple stochastic game, which we call the \emph{probabilistic bisimilarity game}, where the min player tries to show that two states are probabilistic bisimilar, while the max player tries to prove the opposite.

In our probabilistic bisimilarity game, there is a vertex $(s,t)$ for each pair states $s$ and $t$ in $\mathcal{A}$.  If $\ell(s) \neq \ell(t)$ then the vertex $(s, t)$ is a 1-sink.  Otherwise, $(s, t)$ is a min vertex.  In this vertex, the min player selects a set $R \in \mathcal{R}(\delta(s), \delta(t))$ of pairs of transitions.
This set $R$ captures potential matchings of transitions from state $s$ and state $t$.  Subsequently, the max player chooses a pair of transitions from the set $R$.  Once the max player has chosen a pair $(\mu, \nu)$ from the set $R$ corresponding to the transitions $s \rightarrow \mu$ and $t \rightarrow \nu$, the min player can choose a measure-coupling $\omega \in \Omega(\mu,\nu)$.  To ensure that the game graph is finite, we restrict our attention to the vertices $V(\Omega(\mu, \nu))$ of the polytope $\Omega(\mu, \nu)$.  Such a measure-coupling $\omega$ captures a matching of the probability distributions $\mu$ and $\nu$.  Recall that a measure-coupling is a probability distribution on $S \times S$.  From a random vertex $\omega$, the game proceeds to vertex $(u, v)$ with probability $\lambda \cdot \omega(u, v)$ and
to the $0$-sink vertex $\bot$ with probability $1-\lambda$.
Intuitively, the choices of $R \in \mathcal{R}(\delta(s), \delta(t))$ and then $(\mu, \nu) \in R$, performed respectively by the min and the max player, correspond to the $\min$ and $\max$ of Theorem~\ref{th:hausdorff}; analogously, the selection of $\omega \in V(\Omega(\mu, \nu))$ by the min player models the $\min$ in the definition of the Kantorovich lifting.

\medskip
Formally, our probabilistic bisimilarity game for the automaton $\mathcal{A}$ is defined as follows.
\begin{defi}%
\label{def:bssg}
Let $\lambda \in \lopen{0,1}$. The \emph{probabilistic bisimilarity game} $(V, E, P)$ is defined by
\begin{itemize}
\item $V_0 = \{ \bot \}$,
\item $V_1 = \big\{ (s, t) \in S \times S \mid \ell(s) \not= \ell(t) \big\}$,
\item $V_{\max} = \bigcup \big\{ \mathcal{R}(\delta(s), \delta(t)) \mid (s, t) \in V_{\min} \big\}$,
\item $V_{\min} = \big\{ (s, t) \in S \times S \mid \ell(s) = \ell(t) \big\} \cup \bigcup \big\{ R \mid R \in V_{\max} \big\}$,
\item $V_{\mathrm{rnd}} = \bigcup \big\{ V(\Omega(\mu, \nu)) \mid (\mu, \nu) \in V_{\min} \big\}$,
\end{itemize}
\begin{align*}
E &=
 \big\{ ((s, t), R) \mid (s, t) \in V_{\min} \land R \in \mathcal{R}(\delta(s), \delta(t)) \big\} \cup {}  \\
&\phantom{{} = {}} \big\{ (R, (\mu, \nu)) \mid R \in V_{\max} \land (\mu, \nu) \in R \big\} \cup {} \\
&\phantom{{} = {}} \big\{ ((\mu, \nu), \omega) \mid (\mu, \nu) \in V_{\min} \land \omega \in V(\Omega(\mu, \nu)) \big\} \cup {} \\
&\phantom{{} = {}} \big\{ (\omega, (u, v)) \mid \omega \in V_{\mathrm{rnd}} \land (u, v) \in \supp(\omega) \big\} \cup \big\{ (\omega, \bot) \mid \omega \in V_{\mathrm{rnd}} \big\}
\,,
\end{align*}
and, for all $\omega \in V_{\mathrm{rnd}}$ and $(s,t) \in \supp(\omega)$, $P(\omega)((s,t)) = \lambda \cdot \omega(s,t)$ and $P(\omega)(\bot) = 1-\lambda$.
\end{defi}

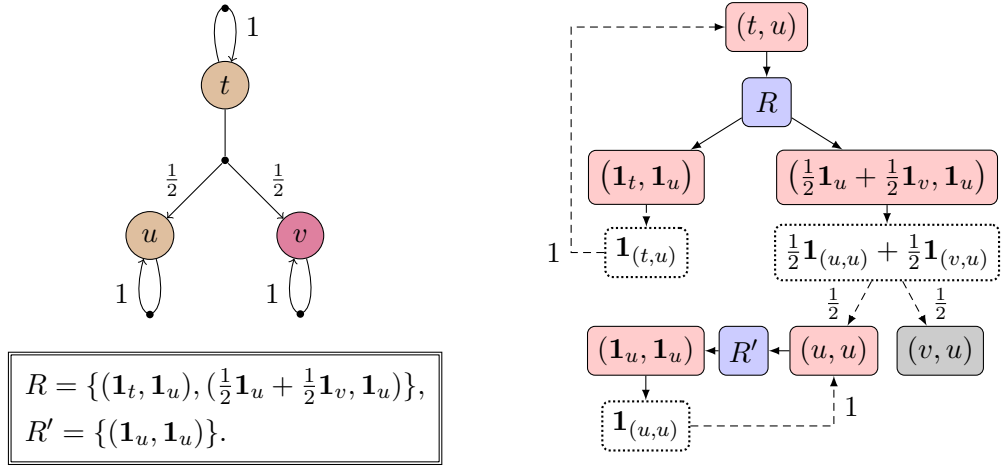
\begin{figure}[t]
\centering

\tikzstyle{vertexbrown} = [vertex, fill=brown!50]
\tikzstyle{squarepurple} = [vertex, fill=purple!50]
\tikzstyle{distribution} = [densely dashed]
\begin{tikzpicture}[baseline=(current bounding box.center)]

\node[vertexbrown] (t) at (2,2) {$t$};
\node[vertexbrown] (u) at (1,0) {$u$};
\node[squarepurple] (v) at (3,0) {$v$};

\node[bullet] (b) at (2,1){};

\draw[-] (t) to node [midway, left] {}  (b);
\draw[->] (b) to node [pos=0.3, left=0.2] {\small $\frac{1}{2}$}  (u);
\draw[->] (b) to node [pos=0.3, right=0.2] {\small $\frac{1}{2}$}  (v);

\path[->] (t) edge [loop above, min distance=10mm] node[bullet, above=-2pt] {}
  node [right, pos=0.8] {$1$} (t); 
\path[->] (u) edge [loop below, min distance=10mm] node[bullet, above=-2pt] {}
node [left, pos=0.8] {$1$} (u); 
\path[->] (v) edge [loop below, min distance=10mm] node[bullet, above=-2pt] {}
node [left, pos=0.8] {$1$} (v); 

\node[draw, double, align=left, inner sep=2mm] at (2,-2.3)
{$\phantom{}R = \{ (\dirac_t, \dirac_u), (\frac{1}{2}\dirac_u + \frac{1}{2}\dirac_v, \dirac_u) \}$, \\[1ex]
$R' = \{ (\dirac_u, \dirac_u) \}$.};
\end{tikzpicture}
\hspace{1cm}
%
\begin{tikzpicture}[node distance=1cm, baseline=(current bounding box.center)]

\node[min] (tu) at (0,0) {$(t,u)$};
\node[max, below of=tu] (R) {$R$};
\node[min, below of=R, xshift=-1.6cm] (r1) {$\big(\dirac_t, \dirac_u\big)$};
\node[min, below of=R, xshift=1.6cm] (r2) {$\big(\frac{1}{2}\dirac_u + \frac{1}{2}\dirac_v, \dirac_u\big)$};

\node[rnd, below of=r1] (omega1)  {$\dirac_{(t,u)}$};
\node[rnd, below of=r2] (omega2)  {$\frac{1}{2}\dirac_{(u,u)} + \frac{1}{2}\dirac_{(v,u)}$};
\node[min, below left of=omega2, yshift=-0.6cm] (uu) {$(u,u)$};
\node[onesink, below right of=omega2, yshift=-0.6cm] (vu) {$(v,u)$};
\node[max, left of=uu, xshift=-2mm] (Rprime) {$R'$};
\node[min, left of=Rprime, xshift=-3mm] (rp) {$(\dirac_u, \dirac_u)$};
\node[rnd, below of=rp] (omega3) {$\dirac_{(u,u)}$};


\path[-latex]
(tu) edge (R)
(R) edge (r1)
      edge (r2)
(r1) edge[distribution] (omega1)
(r2) edge[distribution] (omega2)
(uu) edge (Rprime)
(Rprime) edge (rp)
(rp) edge (omega3);
\draw[-latex, distribution] (omega1) -- +(left:1cm) node[left] {$1$} |- (tu);
\draw[-latex, distribution] (omega3) -| node[above right] {$1$} (uu);
\draw[-latex, distribution] (omega2) -- node[left=1mm] {\small $\frac{1}{2}$} (uu);
\draw[-latex, distribution] (omega2) -- node[right=1mm] {\small $\frac{1}{2}$} (vu);
\end{tikzpicture}

\caption{(Top left:) A probabilistic automaton and (Right:) the associated simple stochastic game constructed
as in Definition~\ref{def:bssg} for $\lambda = 1$ (only the portion reachable
from $(t,u)$ is shown), where $\dirac_x$ denotes the Dirac distribution concentrated at $x$.}%
\label{fig:bisimilaritygame}
\end{figure}

By construction of the probabilistic bisimilarity game, there is a direct correspondence between the function $\Phi$ from Definition~\ref{definition:Phi} associated to the probabilistic bisimilarity game and the function $\Delta_\lambda$ from Section~\ref{sec:pautomata} associated to the probabilistic automaton.  From this correspondence it is straightforward that the respective least fixed points of $\Phi$ and $\Delta_\lambda$ agree, that is, the probabilistic bisimilarity distances of a probabilistic automaton are the values of the corresponding vertices of the probabilistic bisimilarity game.

\begin{thm}%
\label{thm:game}
For all $\lambda \in \lopen{0,1}$ and $s,t \in S$, $\bdist(s,t) = \phi(s,t)$.
\end{thm}
\begin{proof}
The proof is similar to that of~\cite[Theorem~14]{BreugelW14}.

Let $\phi$ be the value function of the probabilistic bisimilarity game.
\change{Since $\Phi$ is monotone and non-expansive (Propositions~\ref{prop:monotonePhi} and~\ref{prop:nonexpPhi}),} we conclude from~\cite[Corollary~1]{Breugel12} that the closure ordinal of $\Phi$ is $\omega$, that is, $\phi$ is the least upper bound of $\{\, \Phi^n(\mathbf{0}) \mid n \in \mathbb{N} \,\}$, where the function $\mathbf{0}$ maps every vertex to zero.  Similarly, $\bdist$ is the least upper bound of $\{\, \Delta_\lambda^n(\mathbf{0}) \mid n \in \mathbb{N} \,\}$, where the function $\mathbf{0}$ maps every state pair to zero.  Therefore, it suffices to show that for all $s,t \in S$ and $n \in \mathbb{N}$,
\begin{equation*}
\Phi^{4n}(\mathbf{0})(s, t) = \Delta_\lambda^n(\mathbf{0})(s, t)
\end{equation*}
by induction on $n$. Obviously, the above holds if $n=0$. Let $n > 0$.  We distinguish the following cases.
\begin{itemize}
\item
If $\ell(s) \not= \ell(t)$ then the vertex $(s, t)$ is a 1-sink and, hence,
$
\Phi^{4n}(\mathbf{0})(s, t)
= 1
= \Delta_\lambda^n(\mathbf{0})(s, t).
$
\item
If $\ell(s) = \ell(t)$ then
\begin{align*}
&\Phi^{4n}(\mathbf{0})(s, t) = \\
&= \min_{R \in \mathcal{R}(\delta(s), \delta(t))} \Phi^{4n-1}(\mathbf{0})(R)\\
&= \min_{R \in \mathcal{R}(\delta(s), \delta(t))} \max_{(\mu, \nu) \in R} \Phi^{4n-2}(\mathbf{0})(\mu, \nu)\\
&= \min_{R \in \mathcal{R}(\delta(s), \delta(t))} \max_{(\mu, \nu) \in R} \min_{\omega \in V(\Omega(\mu, \nu))} \Phi^{4n-3}(\mathbf{0})(\omega)\\
&= \min_{R \in \mathcal{R}(\delta(s), \delta(t))} \max_{(\mu, \nu) \in R} \min_{\omega \in V(\Omega(\mu, \nu))}
 \sum_{(u, v) \in \supp(\omega)} \lambda \omega(u, v) \, \Phi^{4n-4}(\mathbf{0})(u, v) + (1 - \lambda) \Phi^{4n-4}(\mathbf{0})(\bot) \\
&= \min_{R \in \mathcal{R}(\delta(s), \delta(t))} \max_{(\mu, \nu) \in R} \min_{\omega \in V(\Omega(\mu, \nu))}
\lambda \sum_{u, v \in S} \omega(u, v) \, \Phi^{4n-4}(\mathbf{0})(u, v)
\tag{$\bot$ is a $0$-sink}\\
&= \min_{R \in \mathcal{R}(\delta(s), \delta(t))} \max_{(\mu, \nu) \in R} \min_{\omega \in V(\Omega(\mu, \nu))}
\lambda \sum_{u, v \in S} \omega(u, v) \, \Delta_\lambda^{n-1}(\mathbf{0})(u, v)
\tag{by induction}\\
& = \lambda\cdot \min_{R \in \mathcal{R}(\delta(s), \delta(t))} \max_{(\mu, \nu) \in R} \K{\Delta_\lambda^{n-1}(\mathbf{0})}(\mu, \nu)\\
& = \lambda \cdot \H{ {\K{ \Delta_\lambda^{n-1}(\mathbf{0})} } }(\delta(s), \delta(t))
\tag{Theorem~\ref{th:hausdorff}} \\
& = \Delta_\lambda^{n}(\mathbf{0})(s, t) \,. \qedhere
\end{align*}
\end{itemize}
\end{proof}

\noindent
Consider a state pair $(s, t)$ with $s \sim t$.  By Theorem~\ref{th:behaviouraldist}, $\bdist(s, t) = 0$.
Hence, from Theorem~\ref{thm:game} we can conclude that $\phi(s, t) = 0$.  Therefore, by pre-computing probabilistic bisimilarity, $(s, t)$ can be represented as a 0-sink, rather than a min vertex. For example, in Figure~\ref{fig:bisimilaritygame}
this amounts to turning $(u,u)$ into a $0$-sink and disconnecting it from its successors.




Games similar to the above introduced probabilistic bisimilarity game
have been presented in~\cite{DLT08,BreugelW14,FKP17,KM18}.  The game presented by van Breugel and Worrell in~\cite{BreugelW14} is most closely related to our game.  They also consider probabilistic automata and map a probabilistic automaton to a simple stochastic game.  The only difference is that they use the original definition of the Hausdorff distance, whereas we use M\'emoli's alternative characterization.  The games described in~\cite{DLT08,FKP17,KM18} are not stochastic.  Desharnais, Laviolette and Tracol~\cite{DLT08} define an $\epsilon$-probabilistic bisimulation game for probabilistic automata, where $\epsilon > 0$ captures the maximal amount of difference in behaviour that is allowed.  Their measure of difference in behaviour is incomparable to our probabilistic bisimilarity distances (see~\cite[Section~6]{DLT08}).  K\"onig and Mika-Michalski~\cite{KM18} generalize the game of Desharnais et al.\  in a categorical setting so that it is applicable to a large class of systems including probabilistic automata.  Fijalkow, Klin and Panangaden~\cite{FKP17} consider a more restricted class of systems, namely systems with probabilities but without nondeterminism.  In the games in~\cite{DLT08,FKP17,KM18} players choose sets of states, a phenomenon that one does not encounter in our game.

\section{A Coupling Characterisation of the Bisimilarity Distance}%
\label{sec:couplingdistance}

In this section we provide an alternative characterisation for the probabilistic bisimilarity distance
$\bdist$ based on the notion of coupling structure for a probabilistic automaton.
This characterisation generalises the one by Chen et al.~\cite[Theorem 8]{ChenBW12} (see also~\cite[Theorem 8]{BacciLM:tacas13}) for the bisimilarity pseudometric of Desharnais et al.~\cite{DesharnaisGJP04} for labelled Markov chains. Our construction exploits M\'emoli's dual characterisation of the Hausdorff distance (Theorem~\ref{th:hausdorff}).

%
%
\begin{defi}
A \emph{coupling structure} for $\mathcal{A}$ is a tuple $\C = (f, \rho)$ consisting of
\begin{itemize}
\item a map $f \colon \Distr{S} \times \Distr{S} \to \Distr{S \times S}$ such that, for all $\mu, \nu \in \Distr{S}$,  $f(\mu, \nu) \in \Omega(\mu,\nu)$, and
\item a map $\rho \colon S \times S \to 2^{\Distr{S} \times \Distr{S}}$, such that for all $s,t \in S$, $\rho(s,t) \in \mathcal{R}(\delta(s), \delta(t))$.
\end{itemize}
\end{defi}

\noindent
For convenience, the components $f$ and $\rho$ of a coupling structure will be called
\emph{measure-coupling map} and \emph{set-coupling map}, respectively.

\change{The definition of coupling structure is better understood in relation to the
automaton it induces. The probabilistic automaton induced from $\C = (f, \rho)$, denoted
\begin{equation*}
\mathcal{A}_\C = (S \times S, L \times L, \to_\C, \ell_\C) \,,
\end{equation*}
has $S \times S$ as set of states, $L \times L$ as set of labels,  transition relation
${\to_\C} \subseteq (S \times S) \times \Distr{S \times S}$, defined as $(s,t) \to_\C f(\mu,\nu)$ if $(\mu,\nu) \in \rho(s,t)$, and labeling function $\ell_\C \colon S \times S \to L \times L$ defined as
$\ell_\C(s,t) = (\ell(s), \ell(t))$.
Intuitively, $\mathcal{A}_\C$ describes the concurrent execution of two copies of the probabilistic
automaton $\mathcal{A}$, synchronized by the coupling structure $\C$. Coupling structures are used to limiting
the non-determinism only to set-couplings by fixing a particular choice $f(\mu,\nu)$ of measure-coupling between any pair of distributions $\mu,\nu$. }

\medskip
Let $\lambda \in \lopen{0,1}$. For each $\C$ we define the function $\Gamma_\lambda^{\C} \colon {[0,1]}^{S \times S} \to {[0,1]}^{S \times S}$ as
\begin{equation*}
\Gamma_\lambda^{\C}(d)(s,t) =
\begin{cases}
 1 & \text{if $\ell(s) \neq \ell(t)$} \\
 \lambda \cdot \max \{ \sum_{u,v \in S} d(u,v) \cdot \omega(u,v) \mid (s,t) \to_\C \omega \} & \text{otherwise.}
\end{cases}
\end{equation*}

\begin{lem}\label{lem:monoGamma}
The function $\Gamma_\lambda^{\C}$ is well-defined and monotone.
\end{lem}
\begin{proof}
The well definition of $\Gamma_\lambda^{\C}$ follows by the fact that $\lambda \in \lopen{0,1}$ and $\sum_{u,v \in S} d(u,v) \cdot \omega(u,v)$ is a convex combination of a sequence of $[0,1]$-valued numbers, namely ${(d(u,v))}_{u,v \in S}$.

As for monotonicity, let $d, d' \in {[0,1]}^{S \times S}$ and $d \sqsubseteq d'$. Let $s,t \in S$, it suffices to show that $\Gamma_\lambda^{\C}(d)(s,t) \leq \Gamma_\lambda^{\C}(d')(s,t)$. We distinguish the following cases:
\begin{itemize}
\item If $\ell(s) \neq \ell(t)$, then $\Gamma_\lambda^{\C}(d)(s,t) = 1 = \Gamma_\lambda^{\C}(d')(s,t)$.
\item  If $\ell(s) = \ell(t)$, then we have
\begin{align*}
\Gamma_\lambda^{\C}(d)(s,t) &= \lambda \cdot \max \{ \textstyle\sum_{u,v \in S} d(u,v) \cdot \omega(u,v) \mid (s,t) \to_\C \omega \} \\
& = \lambda \textstyle \sum_{u,v \in S} d(u,v) \cdot \omega^*(u,v) \tag{for some $(s,t) \to_\C \omega^*$} \\
& \leq \lambda \textstyle\sum_{u,v \in S} d'(u,v) \cdot \omega^*(u,v) \tag{$d \sqsubseteq d'$ and $\omega^*(u,v) \geq 0$ for all $u,v \in S$} \\
& \leq \lambda \cdot \max \{ \textstyle\sum_{u,v \in S} d'(u,v) \cdot \omega(u,v) \mid (s,t) \to_\C \omega \} \tag{$(s,t) \to_\C \omega^*$} \\
& = \Gamma_\lambda^{\C}(d')(s,t) \,.
\tag*{\qedhere}
\end{align*}
\end{itemize}
\end{proof}
\noindent
By Knaster-Tarski's fixed point theorem, $\Gamma_\lambda^{\C}$ has a least fixed point, denoted by $\gamma^{\C}_\lambda$.
As in~\cite{BacciLM:tacas13}, we call $\gamma^{\C}_\lambda$ the
\emph{$\lambda$-discounted discrepancy} w.r.t.\ $\C$ or simply \emph{$\lambda$-discrepancy}.

\begin{rem}%
\label{rem:discrepancy}
Note that, the \change{$1$-discrepancy} $\gamma_1^{\C}(s,t)$ is the maximal probability of reaching a pair of states $(u,v)$ in
the probabilistic automaton $\mathcal{A}_\C$ such that $\ell(u) \neq \ell(v)$ by starting from the state pair
$(s,t)$. It is well known that the maximal reachability probability can be computed in polynomial-time
as the optimal solution of a linear program (see~\cite[Theorem 10.100]{BaierK08} or~\cite[Chapter~6]{Puterman94}).
The linear program can be trivially generalized to compute $\gamma_\lambda^{\C}$, for any
$\lambda \in \lopen{0,1}$.
\end{rem}

\begin{lem}%
\label{lem:prefixfpoint}
For all $\lambda \in \lopen{0,1}$ and coupling structure $\C$ of $\mathcal{A}$, $\Delta_\lambda(\gamma^{\C}_\lambda) \sqsubseteq \gamma^{\C}_\lambda$.
\end{lem}
\begin{proof}
Let $\C = (f,\rho)$.
Let $s,t \in S$ and $R = \rho(s,t)$. We distinguish two cases.
\begin{itemize}
\item If $\ell(s) \neq \ell(t)$, then $\Delta_\lambda(\gamma^{\C}_\lambda)(s,t) = 1 = \gamma^{\C}_\lambda(s,t)$.
\item If $\ell(s) = \ell(t)$,
\begin{align*}
	\Delta_\lambda(\gamma^{\C}_\lambda)(s,t)
	&= \lambda\cdot \H{\K{{\gamma^{\C}_\lambda}}}(\delta(s), \delta(t)) \tag{def.\ $\Delta_\lambda$} \\
	&= \lambda\cdot \min \{ \textstyle \max_{(\mu, \nu) \in R'} \K{{\gamma^{\C}_\lambda}}(\mu, \nu) \mid R' \in \mathcal{R}(\delta(s), \delta(t)) \} \tag{Theorem~\ref{th:hausdorff} } \\
	&\leq \lambda\cdot \textstyle \max_{(\mu, \nu) \in R} \K{{\gamma^{\C}_\lambda}}(\mu, \nu) \tag{$R \in \mathcal{R}(\delta(s), \delta(t))$} \\
	&= \lambda\cdot \textstyle \max_{(\mu, \nu) \in R} \min_{\omega \in \Omega(\mu,\nu)} \sum_{u,v \in S} \gamma^{\C}_\lambda(u,v) \cdot \omega(u,v) \tag{def. $\K{{\gamma^{\C}_\lambda}}$} \\
	& \leq \lambda\cdot \textstyle \max_{(\mu,\nu) \in R}  \sum_{u,v \in S} \gamma^{\C}_\lambda(u,v) \cdot f(\mu,\nu)(u,v) \tag{$f(\mu, \nu) \in \Omega(\mu,\nu)$} \\
	&= \lambda\cdot \textstyle \max \{ \sum_{u,v \in S} \gamma^{\C}_\lambda(u,v) \cdot \omega(u,v) \mid (s,t) \to_\C \omega \}
	\tag{def.\ $\to_\C$} \\
	&= \Gamma_\lambda^{\C}(\gamma^{\C}_\lambda)(s,t) \tag{def.\ $\Gamma_\lambda^{\C}$} \\
	&= \gamma^{\C}_\lambda(s,t) \,. \tag{$\gamma^{\C}_\lambda$ fixed point of $\Gamma_\lambda^{\C}$}
\end{align*}
\end{itemize}
By the generality of the chosen $s$ and $t$, we conclude that $\Delta_\lambda(\gamma^{\C}_\lambda) \sqsubseteq \gamma^{\C}_\lambda$.
\end{proof}

\begin{cor}\label{cor:prefixfpoint}
For all $\lambda \in \lopen{0,1}$ and coupling structure $\C$ for $\mathcal{A}$, $\bdist \sqsubseteq \gamma^{\C}_\lambda$.
\end{cor}
\begin{proof}
By Knaster-Tarski's fixed point theorem, $\bdist$ is the least prefix point of $\Delta_{\lambda}$, therefore by Lemma~\ref{lem:prefixfpoint}
we can conclude that $\bdist \sqsubseteq \gamma^{\C}_\lambda$.
\end{proof}

The next lemma shows that the probabilistic bisimilarity distance can be characterised as the $\lambda$-discrepancy for a \emph{vertex coupling structure}, that is, a coupling structure $\C = (f,\rho)$ such that $f(\mu,\nu) \in V(\Omega(\mu,\nu))$ for all $\mu, \nu \in \Distr{S}$.
\begin{lem}%
\label{lem:optimalcoupling}
For all $\lambda \in \lopen{0,1}$ there exists a vertex coupling structure $\C$ for $\mathcal{A}$ such that $\bdist = \gamma^{\C}_\lambda$.
\end{lem}
\begin{proof}
We construct a vertex coupling structure $\C = (f,\rho)$ as follows. 

We define $f \colon \Distr{S} \times \Distr{S} \to \Distr{S \times S}$ by
\[\textstyle f(\mu, \nu) \in \argmin_{\omega \in V(\Omega(\mu, \nu))} \sum_{u, v \in S} \bdist(u, v) \cdot \omega(u, v) \,.\]
Hence,
\begin{equation}
 \textstyle \K{\bdist}(\mu,\nu) = \sum_{u,v \in S} \bdist(u,v) \cdot f(\mu,\nu)(u,v) \,.
 \label{eq:effe}
\end{equation}

We define $\rho \colon S \times S \to 2^{\Distr{S} \times \Distr{S}}$ by
\begin{equation*}
\rho(s, t) =  \big\{ \big(\mu, \argmin_{\nu \in \delta(t)} \K{\bdist}(\mu, \nu) \big) \mid \mu \in \delta(s) \big\} \cup
\big\{ \big(\argmin_{\mu \in \delta(s)} \K{\bdist}(\mu, \nu), \nu \big) \mid \nu \in \delta(t) \big\} \,.
\end{equation*}
Hence, $\rho(s, t)\in \mathcal{R}(\delta(s), \delta(t))$ and
\begin{equation}
\mathcal{H}(\K{\bdist})(\delta(s), \delta(t)) = \max \{ \K{\bdist}(\mu, \nu) \mid (\mu, \nu) \in \rho(s, t) \} \,.  \label{eq:erre}
\end{equation}
Next, we show that $\Gamma_\lambda^{\C}(\bdist) \sqsubseteq \bdist$.
Let $s,t \in S$. We distinguish two cases:
\begin{itemize}
\item If $\ell(s) \neq \ell(t)$, then $\bdist(s,t) = \Delta_\lambda(\bdist)(s,t) = 1 = \Gamma_\lambda^{\C}(\bdist)(s,t)$.

\item  If $\ell(s) = \ell(t)$, we have
\begin{align}
\Gamma_\lambda^{\C}(\bdist)(s,t)
&= \lambda \cdot \max \{ \textstyle \sum_{u,v \in S} \bdist(u,v) \cdot \omega(u,v) \mid (s,t) \to_\C \omega \}
	\tag{def.\ $\Gamma_\lambda^{\C}$} \\
&= \lambda \cdot \max \{ \textstyle \sum_{u,v \in S} \bdist(u,v) \cdot f(\mu,\nu)(u,v) \mid (\mu,\nu) \in \rho(s,t) \}
	\tag{def.\ $\to_\C$} \\
&= \lambda \cdot \max \{ \K{\bdist}(\mu, \nu) \mid (\mu, \nu) \in \rho(s,t) \}
	\tag{eq.~\eqref{eq:effe}} \\
&= \lambda \cdot \H{\K{\bdist}}(\delta(s),\delta(t)) \tag{eq.~\eqref{eq:erre}} \\
&= \bdist(s,t) \,. \tag{$\bdist$ fixed point of $\Delta_\lambda$}
\end{align}
\end{itemize}
Therefore $\Gamma_\lambda^{\C}(\bdist) = \bdist$. Since $\gamma^{\C}_\lambda$ is the least fixed point of $\Gamma_\lambda^{\C}$, by Knaster-Tarski's fixed point theorem $\gamma^{\C}_\lambda \sqsubseteq \bdist$.
Moreover, by Corollary~\ref{cor:prefixfpoint}, $\bdist \sqsubseteq \gamma^{\C}_\lambda$. Thus $\bdist = \gamma^{\C}_\lambda$.
\end{proof}

\begin{thm}%
\label{th:mincoupling}
Let $\lambda \in \lopen{0,1}$. Then, the following hold:
\begin{enumerate}
\item\label{th:mincoupling1}
$\bdist = \sqcap \{ \gamma^{\C}_\lambda \mid \text{$\C$ coupling structure for $\mathcal{A}$} \}$;
\item\label{th:mincoupling2}
$s \sim t$ iff $\gamma_\lambda^\C(s,t) = 0$ for some vertex coupling structure $\C$ for $\mathcal{A}$.
\end{enumerate}
\end{thm}
\begin{proof}
\eqref{th:mincoupling1} follows by Corollary~\ref{cor:prefixfpoint} and Lemma~\ref{lem:optimalcoupling};~\eqref{th:mincoupling2} by Theorem~\ref{th:behaviouraldist} and Lemma~\ref{lem:optimalcoupling}. 
\end{proof}

Note that together with Lemma~\ref{lem:optimalcoupling}, Theorem~\ref{th:mincoupling}.\ref{th:mincoupling1}
states that $\bdist$ is the minimal $\lambda$-discrepancy obtained by ranging over the subset of vertex coupling structures.

\begin{rem}[On the relation with probabilistic bisimilarity games]%
\label{rem:couplingstrategies}
The coupling structure characterization of the distance is strongly related to the simple stochastic game
characterization presented in Section~\ref{sec:ssg}. Indeed, the notion of vertex coupling structure
captures essentially the strategies for the min player on a probabilistic bisimilarity game in the following sense: the measure-coupling map component describes the strategy on the vertices of the form $(\mu, \nu) \in R$ for some $R \in V_{\max}$, while the set-coupling map deals with the description of the strategy on the min vertices $(s,t) \in S \times S$.
The discrepancy $\gamma^\C_1$ captures the value w.r.t an optimal strategy for the max player when the min player has fixed their strategy a priori.
\end{rem}

\section{Computing the Bisimilarity Distance}%
\label{sec:algorithm}

We describe a procedure for computing the bisimilarity distances based on Condon's simple
policy iteration algorithm~\cite{Condon90}. Our procedure extends
a similar one proposed in~\cite{TangB16,BacciLM:tacas13} for computing the bisimilarity
distance\franck{s} of Desharnais et al.~\cite{DesharnaisGJP04} for labelled Markov chains. The extension
takes into account the additional presence of nondeterminism in the choice of the transitions.

Condon's simple policy iteration algorithm computes the values of a simple stochastic game
provided that the game is \emph{stopping}, i.e., for each pair of strategies
for the min and max players the token reaches a $0$-sink or $1$-sink vertex with probability one.

As we have shown in Theorem~\ref{thm:game}, the probabilistic bisimilarity distances are the values
of the corresponding vertices in the simple stochastic game given in Definition~\ref{def:bssg}.
Thus, if we prove that the game is stopping we can apply Condon's simple policy iteration
algorithm to compute the probabilistic bisimilarity distances.

\begin{prop}%
\label{prop:stopping}
For $\lambda \in (0,1)$, the simple stochastic game in Definition~\ref{def:bssg} is
stopping.
\end{prop}
\begin{proof}
For each pair of strategies for the min and max players, each vertex in the induced Markov chain
reaches the $0$-sink vertex $\bot$ with probability at least $1 - \lambda$. Since $\lambda < 1$,
from any state, the probability of never reaching $\bot$ is zero, i.e., the probability
of eventually reaching the sink state $\bot$ is one.
\end{proof}

However, for $\lambda = 1$ the game in Definition~\ref{def:bssg} may not be stopping as
shown below.
\begin{exa}%
\label{ex:nonstopping}
Consider the probabilistic automaton in Figure~\ref{fig:bisimilaritygame} and its associated
probabilistic bisimilarity game. By choosing a strategy $\sigma_{\max}$ for the max player
such that $\sigma_{\max}(R) = (\dirac_t, \dirac_u)$, the vertex $(t,u)$ has probability zero
to reach a sink. This can be seen in Figure~\ref{fig:bisimilaritygame}, since there are no
paths using the edge $(R, (\dirac_t, \dirac_u) )$ leading to a sink.
\end{exa}

In~\cite{TangB16}, by imposing the bisimilar state pairs to be $0$-sinks, for the case of labelled Markov
chains the simple stochastic game was proven to be stopping. This method does not generalize to
probabilistic automata. Indeed, Example~\ref{ex:nonstopping} provides a counterexample
even when bisimilar state pairs are $0$-sinks.

In the remainder of the section, we provide a general algorithm to compute the bisimilarity
distance for every $\lambda \in \lopen{0,1}$, by adapting Condon's simple policy iteration algorithm.
Our solution will exploit the coupling characterization of the distance discussed in Section~\ref{sec:couplingdistance}. This allows us to skip the construction of the
simple stochastic game which may have size exponential in the number of states of the automaton.

\subsection{Simple Policy Iteration Strategy}%
\label{sec:simplepolicy}

Condon's algorithm iteratively updates the strategies of the min and max players in turn,
on the basis of the current
over-approximation of the value of the game.
Next we show how Condon's policy updates can be performed directly on coupling structures.

For the update of the coupling structure, we use a measure-coupling map $k(d)(\mu, \nu) \in V(\Omega(\mu,\nu))$ and a set-coupling map $h(d)(s,t) \in \mathcal{R}(\delta(s),\delta(t))$ such that
\begin{align}
k(d)(\mu, \nu) &\in \textstyle
\argmin \Big\{ \sum_{u,v \in S} \omega(u,v) \cdot d(u,v) \mid \omega \in V(\Omega(\mu,\nu)) \Big\} \, \text{, and}
	\label{eq:defk} \\
h(d)(s,t) &\in \textstyle
\argmin  \Big\{ \max_{(\mu, \nu) \in R} \K{d}(\mu,\nu) \mid R \in \mathcal{R}(\delta(s),\delta(t)) \Big\} \,.
	\label{eq:defh}
\end{align}
for $d \colon S \times S \to [0,1]$, $\mu, \nu \in \Distr{S}$, and $s,t \in S$.

The following lemma explains how the above ingredients can be used by the min player to improve
its strategy.
\begin{lem}%
\label{lem:update_improves}
Let $\C = (f,\rho)$. If there exist $s,t \in S$ such that $\Delta_\lambda(\gamma_\lambda^{\C})(s,t) < \gamma_\lambda^{\C}(s,t)$ then, $\gamma_\lambda^{\mathcal{D}} \sqsubset \gamma_\lambda^{\C}$
for a coupling structure $\mathcal{D} = (k(\gamma_\lambda^\C), \rho[(s,t)/R])$, where $R = h(\gamma_\lambda^\C)(s,t)$.
\end{lem}
\begin{proof}
Assume $\Delta_\lambda(\gamma_\lambda^{\C})(s,t) < \gamma_\lambda^{\C}(s,t)$.
Next we show $\Gamma_\lambda^{\mathcal{D}}(\gamma_\lambda^{\C}) \sqsubset \gamma_\lambda^{\C}$.
In particular we prove that $\Gamma_\lambda^{\mathcal{D}}(\gamma_\lambda^{\C})(s,t) < \gamma_\lambda^{\C}(s,t)$ and, for all $(u,v) \neq (s,t)$, $\Gamma_\lambda^{\mathcal{D}}(\gamma_\lambda^{\C})(u,v) \leq \gamma_\lambda^{\C}(u,v)$.

By $\Delta_\lambda(\gamma_\lambda^{\C})(s,t) < \gamma_\lambda^{\C}(s,t)$, we necessarily have $\ell(s) = \ell(t)$. Thus
\begin{align*}
\Delta_\lambda(\gamma_\lambda^{\C})(s,t)
&= \lambda \cdot \H{\K{\gamma_\lambda^{\C}}}(\delta(s),\delta(t))
\tag{$\ell(s) = \ell(t)$ and def.\ $\Delta_\lambda$} \\
&= \lambda\cdot \min \{ \max_{(\mu,\nu) \in R'} \K{\gamma_\lambda^{\C}}(\mu,\nu)
\mid R' \in \mathcal{R}(\delta(s),\delta(t)) \}
\tag{Theorem~\ref{th:hausdorff}} \\
&= \textstyle \lambda\cdot \max_{(\mu, \nu) \in R} \K{\gamma_\lambda^{\C}}(\mu,\nu)  \tag{$R = h(\gamma_\lambda^\C)(s,t)$ and~\eqref{eq:defh}} \\
& = \textstyle \lambda\cdot \max_{(\mu, \nu) \in R} \min_{\omega \in \Omega(\mu, \nu)} \sum_{u,v \in S} \gamma_\lambda^{\C}(u,v) \cdot \omega(u,v) \tag{def. $\K{\gamma_\lambda^{\C}}$}\\
&= \textstyle \lambda\cdot \max_{(\mu, \nu) \in R} \sum_{u,v \in S} \gamma_\lambda^{\C}(u,v) \cdot k(\gamma_\lambda^{\C})(\mu,\nu)(u,v)  \tag{by~\eqref{eq:defk}} \\
&= \Gamma^\mathcal{D}_\lambda(\gamma_\lambda^{\C})(s,t) \,.
\tag{def.\ $\mathcal{D}$ and $\Gamma^\mathcal{D}_\lambda$}
\end{align*}
Therefore, $\Gamma^\mathcal{D}_\lambda(\gamma_\lambda^{\C})(s,t) = \Delta_\lambda(\gamma_\lambda^{\C})(s,t) < \gamma_\lambda^{\C}(s,t)$.

Let $u,v \in S$ such that $(u,v) \neq (s,t)$. We distinguish two cases.
\begin{itemize}
\item If $\ell(u) \neq \ell(v)$, then $\Gamma_\lambda^{\mathcal{D}}(\gamma_\lambda^{\C})(u,v) = 1 = \Gamma_\lambda^{\C}(\gamma_\lambda^{\C})(u,v) = \gamma_\lambda^{\C}(u,v)$.

\item If $\ell(u) = \ell(v)$, then
\begin{align*}
\Gamma_\lambda^{\mathcal{D}}(\gamma_\lambda^{\C})(u,v)
&= \textstyle \lambda\cdot \max_{(\mu, \nu) \in \rho(u,v)} \sum_{x,y \in S} k(\gamma_\lambda^\C)(\mu,\nu)(x,y) \cdot  \gamma_\lambda^\C(x,y)
\tag{def.\ $\Gamma_\lambda^\mathcal{D}$ and $\mathcal{D}$} \\
&\leq \textstyle \lambda\cdot \max_{(\mu, \nu) \in \rho(u,v)} \sum_{x,y \in S} f(\mu,\nu)(x,y) \cdot  \gamma_\lambda^\C(x,y)
\tag{\eqref{eq:defk}, $f(\mu,\nu) \in \Omega(\mu,\nu)$} \\
& = \Gamma_\lambda^{\C}(\gamma_\lambda^{\C})(u,v) = \gamma_\lambda^\C(u,v) \,. \tag{def.\ $\Gamma_\lambda^\C$ and $\gamma_\lambda^\C$}
\end{align*}
\end{itemize}
Thus $\Gamma_\lambda^{\mathcal{D}}(\gamma_\lambda^{\C}) \sqsubset \gamma_\lambda^{\C}$.
By Knaster-Tarski's fixed point theorem, we conclude that $\gamma_\lambda^{\mathcal{D}} \sqsubset \gamma_\lambda^{\C}$.
\end{proof}

Lemma~\ref{lem:update_improves} suggests that $\C = (f,\rho)$ can be improved by replacing the measure-coupling map $f$ with $k(\gamma_\lambda^\C)$ and updating the set-coupling map
$\rho$ at $(s,t)$ with $R = h(\gamma_\lambda^\C)(s,t)$.

Note that a measure-coupling $k(d)(\mu,\nu)$ satisfying~\eqref{eq:defk} can be computed by solving a
linear program and ensuring that the optimal solution is a vertex of the polytope
$\Omega(\mu, \nu)$~\cite{Orlin1985,KleinschmidtS95}.
A set-coupling $h(d)(s,t)$ satisfying~\eqref{eq:defh} is the following:
\begin{equation}
R = \big\{ (\mu, \phi(\mu)) \mid \mu \in \delta(s) \big\} \cup \big\{ (\psi(\nu), \nu) \mid \nu \in \delta(t) \big\}
\in \mathcal{R}(\delta(s),\delta(t)) \,,
\label{eq:computeh}
\end{equation}
where $\phi,\psi$ are such that $\phi(\mu) \in \argmin_{\nu \in \delta(t)} \K{d}(\mu, \nu)$
and $\psi(\nu) \in \argmin_{\mu \in \delta(s)} \K{d}(\mu, \nu)$.
The following lemma justifies our choice of $h(d)(s,t)$.
\begin{lem}%
\label{lem:update}
Let $R$ be as in~\eqref{eq:computeh}. Then
$\H{\K{d}}(\delta(s),\delta(t)) = \max_{(\mu, \nu) \in R} \K{d}(\mu,\nu)$.
\end{lem}
\begin{proof}
By Theorem~\ref{th:hausdorff} and $R \in \mathcal{R}(\delta(s),\delta(t))$, we have
\[\H{\K{d}}(\delta(s), \delta(t)) \leq \textstyle \max_{(\mu, \nu) \in R} \K{d}(\mu, \nu) \,.\]
Hence, it suffices to prove
\begin{enumerate}[label={(\roman*)}]
  \item\label{itm:update1} $\K{d}(\mu,\phi(\mu)) \leq \H{\K{d}}(\delta(s),\delta(t))$, for all $\mu \in \delta(s)$, and
  \item\label{itm:update2} $\K{d}(\psi(\nu),\nu) \leq \H{\K{d}}(\delta(s),\delta(t))$, for all $\nu \in \delta(t)$.
\end{enumerate}
We prove~\ref{itm:update1}. Let $\mu \in \delta(s)$. Then
\begin{align*}
\H{\K{d}}(\delta(s),\delta(t)) &\geq \max_{\mu' \in \delta(s)} \min_{\nu \in \delta(t)} \K{d}(\mu',\nu)
\tag{def. $\H{}$} \\
& \geq \min_{\nu \in \delta(t)} \K{d}(\mu,\nu)
\tag{$\mu \in \delta(s)$} \\
& = \K{d}(\mu,\phi(\mu)) \tag{$\phi(\mu) \in \argmin_{\nu \in \delta(t)} \K{d}(\mu, \nu)$}
\end{align*}
The proof for~\ref{itm:update2} follows similarly.
\end{proof}

\begin{rem}
The update procedure entailed by Lemma~\ref{lem:update_improves} can be
performed in polynomial-time in the size of the probabilistic automaton $\mathcal{A}$.
Indeed, $k(d)(\mu,\nu)$ can be obtained by solving a transportation problem in polynomial time~\cite{Orlin1985,KleinschmidtS95}.
As for $h(d)(s,t)$, one can obtain $\phi(\mu)$ (resp.\ $\psi(\nu)$) by computing $\K{d}(\mu,\nu)$
in polynomial time and selecting the $\nu$ (resp.\ $\mu$) ranging over $\delta(t)$ (resp.\ $\delta(s)$)
that achieves the minimum.
\end{rem}

\subsection{Discounted case.}%
\label{sec:discountedcase}

The simple policy iteration algorithm for computing $\bdist$ in the case $\lambda < 1$
is presented in Algorithm~\ref{alg:bdistlambda}.
The procedure starts by computing an initial vertex coupling structure $\C_0$ (line~\ref{spi:line1}),
\franck{\eg}, by using the North-West corner method in polynomial time (see, \franck{\eg},~\cite[pg.\ 180]{Strayer89}).
Then it continues by iteratively generating a sequence $\C_0, \C_1, \dots, \C_n$ of vertex coupling structures where $\bdist = \gamma_\lambda^{\C_n}$.
At each iteration, the current coupling structure $\C_i$ is tested for optimality (line~\ref{spi:line2})
by checking whether the corresponding $\lambda$-discrepancy $\gamma_\lambda^{\C_i}$ is a
fixed point for $\Delta_\lambda$. If there exists $(s,t) \in S$ violating the equality
$\gamma_\lambda^{\C_i} = \Delta_\lambda(\gamma_\lambda^{\C_i})$, it constructs $\C_{i+1}$ by updating $\C_i$ at $(s,t)$ as prescribed by Lemma~\ref{lem:update_improves}
(line~\ref{spi:line4}). This guarantees that $\gamma_{\lambda}^{\C_i} \sqsupset \gamma_{\lambda}^{\C_{i+1}}$, i.e., a strict improvement of the $\lambda$-discrepancy towards the minimal one.


Termination follows by the fact that there are only finitely many vertex coupling structures for $\mathcal{A}$. Furthermore, the correctness of the output of the algorithm is due to the fact that, $\Delta_\lambda$ has a unique fixed point when $0 \leq \lambda < 1$.

\begin{algorithm}[t]
\DontPrintSemicolon
Initialise $\C = (f,\rho)$ as an arbitrary vertex coupling structure for $\mathcal{A}$\; \label{spi:line1}
\While{$\exists (s,t).\, \Delta_\lambda(\gamma_\lambda^\C)(s,t) < \gamma_\lambda^\C(s,t)$}{ \label{spi:line2}
	$R \gets h(\gamma_\lambda^\C)(s,t)$\;
	$\C \gets \big(k(\gamma_\lambda^\C), \, \rho[(s,t) / R] \big)$ \tcc*[r]{update coupling structure} \label{spi:line4}
}
\Return $\gamma_\lambda^\C$ \tcc*[r]{$\gamma_\lambda^\C = \bdist$}

\caption{Simple policy iteration algorithm computing $\bdist$ for $\lambda \in (0,1)$.}
\label{alg:bdistlambda}
\end{algorithm}

\begin{thm}%
\label{th:alg:bdistlambda}
Let $\lambda \in (0,1)$. Algorithm~\ref{alg:bdistlambda} is terminates and computes $\bdist$.
\end{thm}
\begin{proof}
First we prove termination. Note that the set
\begin{equation}
	\{\gamma_\lambda^\C \mid \text{$\C$ vertex coupling structure for $\mathcal{A}$}\}
	\label{eq:vetexdiscrepancieslamb}
\end{equation}
is finite because for all $s,t \in S$ the set $\mathcal{R}(\delta(s),\delta(t))$ is finite, and for all $\mu \in \delta(s)$ and $\nu \in \delta(t)$ the polytope $\Omega(\mu,\nu)$ has finitely many vertices, i.e., $V(\Omega(\mu,\nu))$ is finite.
Towards a contradiction, assume that Algorithm~\ref{alg:bdistlambda} does not terminate. Let $\C_0, \C_1, \C_2, \dots$ be the infinite sequence of coupling structures generated during a non-terminating execution of Algorithm~\ref{alg:bdistlambda}. Since the set in~\eqref{eq:vetexdiscrepancieslamb} is finite, there must be $i < j$ such that $\gamma_\lambda^{\C_i} = \gamma_\lambda^{\C_j}$.

On the contrary, next we prove that the updates of the coupling structures in Algorithm~\ref{alg:bdistlambda} ensure that for all $n \in \mathbb{N}$, $\gamma_\lambda^{\C_n} \sqsupset \gamma_\lambda^{\C_{n+1}}$. Let $n \in \mathbb{N}$.
Since we are considering a non-terminating execution we have that $\Delta_\lambda(\gamma_\lambda^{\C_n})(s,t) < \gamma_\lambda^{\C_n}(s,t)$, for some $s,t \in S$. $\C_{n+1}$ is obtained from the update performed in line~\ref{spi:line4}, which is exactly the one prescribed by Lemma~\ref{lem:update_improves}. Therefore we have that $\gamma_\lambda^{\C_n} \sqsupset \gamma_\lambda^{\C_{n+1}}$. Hence, Algorithm~\ref{alg:bdistlambda} terminates.

When the execution of Algorithm~\ref{alg:bdistlambda} has reached the return statement, we have that $\Delta_\lambda(\gamma_\lambda^{\C_n})(s,t) \geq \gamma_\lambda^{\C_n}(s,t)$ for all $s,t \in S$, i.e., $\gamma_\lambda^{\C_n} \sqsubseteq \Delta_\lambda(\gamma_1^{\C_n})$. By Lemma~\ref{lem:prefixfpoint}, $\gamma_\lambda^{\C_n} \sqsupseteq \Delta_\lambda(\gamma_\lambda^{\C_n})$, therefore $\gamma_\lambda^{\C_n} = \Delta_\lambda(\gamma_\lambda^{\C_n})$. The operator $\Delta_\lambda$ is $\lambda$-Lipschitz continuous \change{\cite[Proposition~10.3.2(b)]{T18}} thus, by Banach's fixed-point theorem, $\Delta_\lambda$ has a unique fixed point. Hence, $\gamma_\lambda^\C = \bdist$.
\end{proof}

\subsection{Undiscounted case.}%
\label{sec:undiscountedcase}

For $\lambda = 1$, the termination condition of the simple policy-iteration algorithm of Section~\ref{sec:discountedcase} is not sufficient to guarantee correctness, since Algorithm~\ref{alg:bdistlambda} may terminate prematurely
by returning a fixed point of $\Delta_1$ that is not the minimal one.

Towards a way to obtain a stronger termination condition, we introduce the notion of self-closed relations
w.r.t.\ a fixed point for $\Delta_1$, originally due to~\cite{Fu12}.

\begin{defi}%
\label{def:selfclosed}
A relation $M \subseteq S \times S$ is \emph{self-closed w.r.t.\ $d = \Delta_1(d)$} if, whenever
$s \mathrel{M} t$,
\begin{enumerate}[label={(\roman*)}]
  \item\label{itm:i} $\ell(s) = \ell(t)$ and $d(s,t) > 0$,
  \item\label{itm:ii} if $s \to \mu$ and $d(s,t) = \min_{\nu' \in \delta(t)} \K{d}(\mu,\nu')$ then
  there exists $t \to \nu$ and $\omega \in \Omega(\mu, \nu)$ such that
  $d(s,t) = \sum_{u,v \in S} d(u,v) \cdot \omega(u,v)$ and $\supp(\omega) \subseteq M$,
  \item\label{itm:iii} if $t \to \nu$ and $d(s,t) = \min_{\mu' \in \delta(s)} \K{d}(\mu',\nu)$ then
  there exists $s \to \mu$ and $\omega \in \Omega(\mu, \nu)$ such that
  $d(s,t) = \sum_{u,v \in S} d(u,v) \cdot \omega(u,v)$ and $\supp(\omega) \subseteq M$.
\end{enumerate}
Two states are self-closed w.r.t\ $d$, written $s \approx_d t$, if they are related by some self-closed
relation w.r.t.\ $d$.
\end{defi}

It can be easily shown that $\approx_d$ is the largest self-closed relation w.r.t.\ $d$. Note that
the concept of self-closeness above is defined only for fixed points of $\Delta_1$.
As remarked in~\cite{Fu12}, the largest self-closed relation $\approx_d$ can be computed in polynomial
time by using partition refinement techniques similar to those employed to compute the largest bisimilarity
relation.

\change{
\begin{exa}%
\label{ex:self-closedExample}
The following is a parametric definition for a family of fixed points of $\Delta_1$ over the probabilistic automaton in Figure~\ref{fig:bisimilaritygame}, for $\frac{1}{2} \leq \alpha \leq 1$:
\begin{align*}
& d_\alpha(t,u) = d_\alpha(u,t)= \alpha \,, \\
& d_\alpha(t,t) = d_\alpha(u,u) = d_\alpha(v,v) = 0 \,, \\
& d_\alpha(t,v) = d_\alpha(u,v) = d_\alpha(v,t) = d_\alpha(v,u) = 1\,.
\end{align*}
For $\alpha \neq \frac{1}{2}$, an example of a self-closed relation w.r.t.\ $d_\alpha$ is given by the relation
$M = \{(t,u)\}$. It is easy to verify that condition
(i) holds true since $\ell(t) = \ell(u)$ and $d_\alpha(t,u) = \alpha > 0$. The only case where (ii) is non-trivial is when $t \to \dirac_t$, as $\min_{\nu' \in \delta(u)} \K{d_\alpha}(\dirac_t,\nu') =  \K{d_\alpha}(\dirac_t,\dirac_u) = d_\alpha(t,u)$. In this case, condition (ii) is satisfied since $u \to \dirac_u$ and $\omega = \dirac_{(t,u)} \in \Omega(\dirac_t,  \dirac_u)$ is a measure-coupling such that $d_\alpha(s,t) = \sum_{u,v} d_\alpha(u,v) \cdot \omega(u,v)$ and $\supp(\omega) \subseteq M$. As for the last condition (iii), it holds true because the premise of the implication is never satisfied.
With similar arguments one can easily verify that $M' =  \{(t,u), (u,t)\}$ is also a self-closed set w.r.t.\ $d_\alpha$, and in particular it is the greatest one.

For $\alpha = \frac{1}{2}$, we have that $d_\alpha$ is the least fixed points. The only admissible self-closed relation for it is the empty set.
\end{exa}
}


\change{The next lemma (Lemma~\ref{lem:self-closedGame}) characterizes self-closed relation $M$ in terms of the existence of certain optimal strategies for the probabilistic bisimilarity game given in Definition~\ref{def:bssg}. Intuitively, $M$ is a set of nodes such that, if you play optimally with respect to the current value $d$ you still remain within $M$.}


\change{
\begin{exa}
As shown in Example~\ref{ex:self-closedExample}, for $\frac{1}{2} < \alpha \leq 1$, $M = \{(t,u)\}$ is a self-closed relation w.r.t. $d_\alpha$ for the probabilistic automaton depicted in Figure~\ref{fig:bisimilaritygame}. Consider now $d_\alpha$ as the current value for the probabilistic bisimilarity game associated with the automaton (see Figure~\ref{fig:bisimilaritygame}(right)). An optimal strategy relative to the current value $d_\alpha$ is the one where the $\min$ player chooses $R$ from $(t,u)$, and $\dirac_{(t,u)}$ from $(\dirac_t, \dirac_u)$; while the max player chooses $(\dirac_t, \dirac_u)$ from $R$. This particular choice of the strategy makes the two players stay within $M$, without ever reaching a sink state. 
\end{exa}
}

\change{
\begin{lem}%
\label{lem:self-closedGame}
Given $d = \Delta_1(d)$, a relation $M \subseteq S \times S$ is self-closed w.r.t.\ $d$, if and only if, whenever
$s \mathrel{M} t$ then
\begin{enumerate}[label={(\alph*)}]
  \item\label{itm:a} $\ell(s) = \ell(t)$ and $d(s,t) > 0$,
  \item\label{itm:B} there exists $R \in \mathcal{R}(\delta(s), \delta(t))$ such that $d(s,t) = \max_{(\mu', \nu') \in R} \K{d}(\mu', \nu')$ and for all $(\mu, \nu) \in R$ such that $d(s,t) = \K{d}(\mu, \nu)$, there exists $\omega \in \Omega(\mu, \nu)$ such that $d(s,t) = \sum_{u,v \in S} d(u,v) \cdot \omega(u,v)$ and $\supp(\omega) \subseteq M$.
\end{enumerate}
\end{lem}
\begin{proof}
$(\Leftarrow)$ Let $s \mathrel{M} t$. Def.~\ref{def:selfclosed}.\ref{itm:i} follows from~\ref{itm:a}. Next we show that~\ref{itm:B} implies Def.~\ref{def:selfclosed}.\ref{itm:ii}. Let $R$ be the set-coupling for $(\delta(s), \delta(t))$ satisfying~\ref{itm:B} and $\mu \in \delta(s)$ be such that $d(s,t) = \min_{\nu' \in \delta(t)} \K{d}(\mu,\nu')$. Then by definition of set-coupling, there exists $\nu \in \delta(t)$ such that $(\mu, \nu) \in R$.
Clearly $\K{d}(\mu, \nu) \geq \min_{\nu' \in \delta(t)} \K{d}(\mu,\nu')$. Moreover,
\begin{align*}
\K{d}(\mu, \nu) &\leq \max_{(\mu',\nu') \in R} \K{d}(\mu', \nu') \tag{$(\mu, \nu) \in R$} \\
&= d(s,t) \tag{by~\ref{itm:B}} \\
&= \min_{\nu' \in \delta(t)} \K{d}(\mu,\nu') \,. \tag{hp.\ on $\mu$}
\end{align*}
Therefore, $\K{d}(\mu, \nu) = d(s,t)$. By $(\mu, \nu) \in R$ and~\ref{itm:B}, there exists $\omega \in \Omega(\mu, \nu)$ such that $d(s,t) = \sum_{u,v \in S} d(u,v) \cdot \omega(u,v)$ and $\supp(\omega) \subseteq M$. Hence, Def.~\ref{def:selfclosed}.\ref{itm:ii} holds true.
Condition Def.~\ref{def:selfclosed}.\ref{itm:iii} follows similarly.

$(\Rightarrow)$ Let $s \mathrel{M} t$. Condition~\ref{itm:a} follows by Def.~\ref{def:selfclosed}.\ref{itm:i}.
Assume Def.~\ref{def:selfclosed}.\ref{itm:ii} and~\ref{itm:iii} hold true. Then, we can define $\phi \colon \delta(s) \to \delta(t)$ and $\psi \colon \delta(t) \to \delta(s)$ such that
\begin{enumerate}
\item $\phi(\mu) \in \argmin_{\nu' \in \delta(t)} \K{d}(\mu, \nu')$ and if $\K{d}(\mu, \phi(\mu)) = d(s,t)$ then, there exists $\omega \in \Omega(\mu, \phi(\mu))$ such that $\sum_{u,v \in S} d(u,v) \cdot \omega(u,v) = d(s,t)$ and $\supp(\omega) \subseteq M$.
\item $\psi(\nu) \in \argmin_{\mu' \in \delta(s)} \K{d}(\mu',\nu)$ and if $\K{d}(\psi(\nu),\nu) = d(s,t)$ then there exists $\omega \in \Omega(\psi(\nu),\nu)$ such that $\sum_{u,v \in S} d(u,v) \cdot \omega(u,v) = d(s,t)$ and $\supp(\omega) \subseteq M$.
\end{enumerate}
Define $R = \{ (\mu, \phi(\mu)) \mid \mu \in \delta(s) \} \cup \{ (\psi(\nu), \nu) \mid \nu \in \delta(t) \}$. By construction $R \in \mathcal{R}(\delta(s), \delta(t))$. Then, the following hold
\begin{align*}
\max_{(\mu,\nu) \in R} \K{d}(\mu, \nu)
&= \max\{ \max_{\mu \in \delta(s)} \K{d}(\mu, \phi(\mu)), \max_{\nu \in \delta(t)} \K{d}(\psi(\nu), \nu) \} \tag{def.\ $R$} \\
&= \H{\K{d}}(\delta(s),\delta(t)) \tag{by def.\ $\phi$ and $\psi$} \\
&= d(s,t) \,. \tag{$d = \Delta_1(d)$}
\end{align*}
Let $(\mu, \nu) \in R$ such that $d(s,t) = \K{d}(\mu, \nu)$. By definition, $\nu = \phi(\mu)$ or $\mu = \psi(\nu)$. If $\nu = \phi(\mu)$ then, $d(s,t) = \K{d}(\mu, \phi(\mu))$ and, by $\phi$, there exists $\omega \in \Omega(\mu,\phi(\mu))$ such that $\sum_{u,v \in S} d(u,v) \cdot \omega(u,v) = d(s,t)$ and $\supp(\omega) \subseteq M$. The case $\mu = \psi(\nu)$ is analogous.
\end{proof}
}

The next lemma states that if for a fixed point $d = \Delta_1(d)$ the relation $\approx_d$ is nonempty,
then $d$ is not the least fixed point of $\Delta_1$.
\begin{lem}%
\label{lem:decreasefixpoint}
Let $d = \Delta_1(d)$. If there exists a nonempty self-closed relation $M$ w.r.t.\ $d$, then
there exists $d_M \sqsubset d$ such that $\Delta_1(d_M) \sqsubseteq d_M$.
Moreover, $d_M$ can be computed in polynomial time in the size of the probabilistic automaton $\mathcal{A}$.
\end{lem}
\begin{proof}
Let $M$ be a nonempty self-closed relation w.r.t.\ $d$.
For arbitrary $s,t \in S$, $\mu \in \delta(s)$, and $\nu \in \delta(t)$, define
\begin{equation*}
\theta_s(\mu, t) = d(s,t) - \min_{\nu \in \delta(t)} \K{d}(\mu,\nu)
\qquad \text{and} \qquad
\theta_t(s, \nu) = d(s,t) - \min_{\mu \in \delta(s)} \K{d}(\mu,\nu) \,.
\end{equation*}
Note that, $\theta_s(\mu, t)$ and $\theta_t(s, \nu)$ are non-negative since $d = \Delta_1(d)$.
Let $\theta = \min \{ \theta_1, \theta_2, \theta_3 \}$ where
\begin{itemize}
\item $\theta_1 = \min \{ \theta_s(\mu, t) \mid (s,t) \in M \land \mu \in \delta(s) \land \theta_s(\mu, t) > 0 \}$;
\item $\theta_2 = \min \{ \theta_t(s, \nu) \mid (s,t) \in M \land \nu \in \delta(t) \land \theta_t(s, \nu) > 0 \}$;
\item $\theta_3 = \min \{ d(s,t) \mid (s,t) \in M \}$;
\end{itemize}
where $\min \emptyset = 1$. Note that $\theta_3 > 0$, because $M$ is a nonempty self-closed relation w.r.t. $d$. Therefore $\theta > 0$.
We define the map $d_M \colon S \times S \to [0,1]$ as
\begin{equation*}
d_M(s,t) =
\begin{cases}
  d(s,t) - \theta & \text{if $(s,t) \in M$} \\
  d(s,t) &\text{if $(s,t) \notin M$}
\end{cases}
\end{equation*}
It is clear that $d_M$ is well-defined. Moreover $d_M \sqsubset d$ because $M$ is nonempty and $\theta > 0$.

Next we prove that $ \Delta_1(d_M) \sqsubseteq d_M$.
Let $s,t \in S$. We consider two cases:
\begin{itemize}

\item Assume $(s,t) \notin M$. Then
\begin{align*}
\Delta_1(d_M)(s,t)
&\leq \Delta_1(d)(s,t)  \tag{by $d_M \sqsubseteq d$ and $\Delta_1$ monotone} \\
&= d(s,t)  \tag{$d = \Delta_1(d)$} \\
&= d_M(s,t) \tag{$(s,t) \notin M$}
\end{align*}

\item Assume $(s,t) \in M$. Then $\ell(s) = \ell(t)$. Let $\mu \in \delta(s)$. We consider two subcases below:
\begin{enumerate}
\item\label{thetaGreaterThanZero}
If $\theta_s(\mu, t) > 0$ we have
\begin{align*}
d_M(s,t)
&= d(s,t) - \theta    \tag{def.\ $d_M$} \\
&\geq d(s,t) - \theta_s(\mu, t)   \tag{$0 < \theta \leq \theta_s(\mu, t)$} \\
&= \min_{\nu \in \delta(t)} \K{d}(\mu,\nu)  \tag{def.\ $\theta_s(\mu, t)$}  \\
&\geq \min_{\nu \in \delta(t)} \K{d_M}(\mu,\nu)  \tag{$d_M \sqsubseteq d$ and $\K{}$ monotone}
\end{align*}
\item\label{thetaEqualZero}
If $\theta_s(\mu, t) = 0$, then $d(s,t) = \min_{\nu \in \delta(t)} \K{d}(\mu,\nu)$.
Since $M$ is self-closed w.r.t.\ $d$, there exists
$\nu' \in \delta(t)$ such that $d(s,t) = \sum_{u,v \in S} d(u,v) \cdot \omega(u,v)$, for some
$\omega \in \Omega(\mu,\nu')$ such that $\supp(\omega) \subseteq M$. Thus
\begin{align*}
 \min_{\nu \in \delta(t)} \K{d_M}(\mu, \nu) &\leq \K{d_M}(\mu, \nu') \tag{$\nu' \in \delta(t)$} \\
&=  \min_{\omega' \in \Omega(\mu, \nu')} \sum_{u, v \in S} d_M(u, v) \cdot \omega'(u, v)  \tag{def. $\K{d_M}$} \\
& \leq \sum_{u,v \in S} d_M(u,v) \cdot \omega(u,v) \tag{$\omega \in \Omega(\mu,\nu')$} \\
&= \sum_{u,v \in M} d_M(u,v) \cdot \omega(u,v)  \tag{$\supp(\omega) \subseteq M$} \\
&= \sum_{u,v \in M} (d(u,v) - \theta) \cdot \omega(u,v)  \tag{def.\ $d_M$} \\
&= \Big(\sum_{u,v \in M} d(u,v) \cdot \omega(u,v) \Big) - \theta \tag{$\sum_{u,v \in M} \omega(u,v) = 1$} \\
&= d(s,t) - \theta \tag{$d(s,t) = \sum_{u,v \in S} d(u,v) \cdot \omega(u,v)$} \\
&= d_M(s,t)  \tag{def.\ $d_M$}
\end{align*}
\end{enumerate}
So that, in both cases~\ref{thetaGreaterThanZero} and~\ref{thetaEqualZero} we have
$d_M(s,t) \geq \min_{\nu \in \delta(t)} \K{d_M}(\mu,\nu)$. Since this inequality holds for all $\mu \in \delta(s)$,
we have $d_M(s,t) \geq \max_{\mu \in \delta(s)} \min_{\nu \in \delta(t)} \K{d_M}(\mu,\nu)$.
Symmetrically, we can prove $d_M(s,t) \geq \max_{\nu \in \delta(t)} \min_{\mu \in \delta(s)} \K{d_M}(\mu,\nu)$.
Thus, by definition of Hausdorff lifting, $d_M(s,t) \geq \H{\K{d_M}}(\delta(s),\delta(t))$.
From this we conclude
\begin{align*}
d_M(s,t)
&\geq \H{\K{d_M}}(\delta(s),\delta(t)) \\
&= \Delta_1(d_M)(s,t)  \tag{$\ell(s) = \ell(t)$ and def.\ $\Delta_1$}
\end{align*}
\end{itemize}
Finally, we consider the complexity of computing $d_M$.
For computing $\theta$, we need to compute in turn $\theta_1$, $\theta_2$, and
$\theta_3$.
Since $M \subseteq S \times S$, computing $\theta_3$ can be done in quadratic time in $|S|$.
The computation of $\theta_1$ requires at most $|M| \cdot \sum_{s \in S}|\delta(s)|$ solutions of
a transportation problem. This can be done in polynomial-time in the size of $\mathcal{A}$. Similarly for $\theta_2$.
\end{proof}

The proof of Lemma~\ref{lem:decreasefixpoint} is essentially that of~\cite[Theorem~3]{Fu12}.
Given a nonempty self-closed relation $M$ w.r.t.\ $d$, the above result can be used to obtain a
prefix point of $\Delta_1$, namely $d_M$, that improves $d$ towards the search of the
least fixed point.
The prefix point $d_M$ of Lemma~\ref{lem:decreasefixpoint} is obtained from $d$ by subtracting a
suitable value $\theta > 0$ from all the distances computed at pairs of states in $M$:
\begin{equation*}
d_M(s,t) =
\begin{cases}
  d(s,t) - \theta & \text{if $(s,t) \in M$} \,, \\
  d(s,t) &\text{if $(s,t) \notin M$} \,.
\end{cases}
\end{equation*}
\change{The value of $\theta$ that gives us the smallest prefix point defined as above, is the maximal value satisfying the following inequalities
\begin{align*}
\theta &\leq d(s,t) - \min_{\nu' \in \delta(t)} \K{d}(\mu,\nu')
	&& \text{for all $(s,t) \in M$ and $\mu \in \delta(s)$,}  \\
\theta &\leq d(s,t) - \min_{\mu' \in \delta(s)} \K{d}(\mu',\nu)
	&& \text{for all $(s,t) \in M$ and $\nu \in \delta(t)$,}  \\
\theta &\leq d(s,t)   &&\text{for all $(s,t) \in M$} \,.
\end{align*}
The fact that $d_M$ is a prefix point follows by the fact that $M$ is a self-closed relation.}

\medskip

The following lemma provides us with a termination condition for the simple policy iteration algorithm
to compute $\bdist[1]$. Indeed, according to it, if $d$ is a fixed point of $\Delta_1$, we
can assert that $d$ is equal to bisimilarity distance $\bdist[1]$ by simply checking that the maximal self-closed
relation w.r.t.\ $d$ is empty.
\begin{lem}%
\label{lem:termination}
\change{Let $d = \Delta_1(d)$. If} ${\approx_d} = \emptyset$, then $d = \bdist[1]$.
\end{lem}
\begin{proof}
Let $d = \Delta_1(d)$. We proceed by contraposition. Assume that $d \neq \bdist[1]$. We define a non-empty self-closed relation $M$ w.r.t.\ $d$ as follows.
\begin{align*}
	m = \textstyle \max_{s,t \in S} d(s,t) - \bdist[1](s,t) \,,  && M = \{(s,t) \in S \times S \mid d(s,t) - \bdist[1](s,t) = m \} \, .
\end{align*}
Clearly, $m > 0$ and $M \neq \emptyset$ because $d \neq \bdist[1]$.

Let $(s,t) \in M$. We prove that the three conditions of Definition~\ref{def:selfclosed} hold true.
\begin{enumerate}
\item\label{itm:condition1} $d(s,t) > 0$ because $0 < m = d(s,t) - \bdist[1](s,t) \leq d(s,t)$. Now we prove that $\ell(s) = \ell(t)$. Towards a contradiction, assume $\ell(s) \neq \ell(t)$. Then, the following inequalities hold
\begin{equation*}
0 < m = d(s,t) - \bdist[1](s,t) = \Delta_1(d)(s,t) - \Delta_1(\bdist[1])(s,t) = 1 - 1 = 0 \,,
\end{equation*}
leading to the contradiction that $0 < 0$.
\item Let $\mu \in \delta(s)$ such that $d(s,t) = \min_{\nu \in \delta(t)} \K{d}(\mu, \nu)$. Then we have
\begin{align*}
	\bdist[1](s,t) &= \Delta_1(\bdist[1])(s,t) \tag{by def.\ $\bdist[1]$} \\
	&= \H{\K{\bdist[1]}}(\delta(s),\delta(t)) \tag{by $\ell(s) = \ell(t)$} \\
	&\geq \textstyle \min_{\nu \in \delta(t)} \K{\bdist[1]}(\mu, \nu) \tag{by def.\ $\H{}$}
\end{align*}
Let $\nu^* \in \delta(t)$, $\omega \in \Omega(\mu, \nu^*)$ such that
\begin{equation}
\textstyle \min_{\nu \in \delta(t)} \K{\bdist[1]}(\mu, \nu) = \K{\bdist[1]}(\mu, \nu^*) =  \sum_{u,v \in S} \bdist[1](u,v) \cdot \omega(u,v) \,. \label{eq:omega}
\end{equation}
Then, the following inequalities hold
\begin{align*}
	&\K{\bdist[1]}(\mu,\nu^*) = \textstyle \sum_{u,v \in S} \bdist[1](u,v) \cdot \omega(u,v)
	\tag{by~\eqref{eq:omega}}\\
	&= \textstyle \sum_{u,v \in S} \big(d(u,v) - (d(u,v) - \bdist[1](u,v))\big) \cdot \omega(u,v) \\
	&\geq \textstyle \sum_{u,v \in S} (d(u,v) - m) \cdot \omega(u,v)  \tag{def.\ $m$} \\
	&= \textstyle \big( \sum_{u,v \in S} d(u,v) \cdot \omega(u,v) \big) - m \tag{$\omega \in \Distr{S \times S}$} \\
	&\geq \K{d}(\mu,\nu^*) - m \tag{def.\ $\K{}$ and $\omega \in \Omega(\mu, \nu^*)$}
\end{align*}
Thus, we have
\begin{align*}
	d(s,t) &\leq \K{d}(\mu, \nu^*) \tag{$d(s,t) = \min_{\nu \in \delta(t)} \K{d}(\mu, \nu)$} \\
	& \leq \K{\bdist[1]}(\mu, \nu^*) + m \tag{$\K{\bdist[1]}(\mu,\nu^*) \geq \K{d}(\mu,\nu^*) - m$}\\
	& \leq \bdist[1](s,t) + m \tag{$\bdist[1](s,t) \geq \min_{\nu \in \delta(t)} \K{\bdist[1]}(\mu, \nu)$ and def.\ $\nu^*$} \\
	& = \bdist[1](s,t) + (d(s,t) - \bdist[1](s,t)) \tag{def.\ $m$ and $(s,t) \in M$} \\
	& = d(s,t)
\end{align*}
Therefore, all the above inequalities are, in fact, equalities. Hence,
\begin{align}
d(s,t) = \K{d}(\mu, \nu^*) &&\text{and}&& \bdist[1](s,t) = \K{\bdist[1]}(\mu, \nu^*) \,. \label{eq:nustar}
\end{align}

We conclude by proving that $\omega$ satisfies the following
\begin{align*}
	d(s,t) = \textstyle \sum_{u,v \in S} d(u,v) \cdot \omega(u,v) &&\text{and}&& \supp(\omega) \subseteq M \,.
\end{align*}

This can be observed as follows
\begin{align*}
& \textstyle \sum_{u,v \in S} d(u,v) \cdot \omega(u,v) \geq d(s,t)
	\tag{$\omega \in \Omega(\mu, \nu^*)$ and~\eqref{eq:nustar}} \\
& = \bdist[1](s,t) + m
	\tag{$(s,t) \in M$ and def.\ $m$} \\
& = \big( \textstyle \sum_{u,v \in S} \bdist[1](u,v) \cdot \omega(u,v) \big) + m
	\tag{by~\eqref{eq:nustar} and~\eqref{eq:omega}} \\
& = \textstyle \sum_{u,v \in S} (\bdist[1](u,v) + m) \cdot \omega(u,v)
	\tag{$\omega \in \Distr{S \times S}$} \\
& \geq \textstyle \sum_{u,v \in S} \big(\bdist[1](u,v) + (d(u,v) - \bdist[1](u,v)) \big) \cdot \omega(u,v)
	\tag{def.\ $m$} \\
& = \textstyle \sum_{u,v \in S} d(u,v) \cdot \omega(u,v)
\end{align*}
Hence, the above are in fact equalities and in particular $d(s,t) = \sum_{u,v \in S} d(u,v) \cdot \omega(u,v)$.

Consider now the following inequalities
\begin{align*}
	m &= d(s,t) - \bdist[1](s,t) \tag{$(s,t) \in M$} \\
	&= d(s,t) - \textstyle \sum_{u,v \in S} \bdist[1](u,v) \cdot  \omega(u,v)
		\tag{by~\eqref{eq:nustar} and~\eqref{eq:omega}} \\
	&= \textstyle  \sum_{u,v \in S} \big(d(u,v) - \bdist[1](u,v) \big) \cdot \omega(u,v)
	\tag{$d(s,t) = \sum_{u,v \in S} d(u,v) \cdot \omega(u,v)$} 
\end{align*}
Since $d(u,v) - \bdist[1](u,v) \leq m$ for all $u,v \in S$, the above equalities imply that whenever $\omega(u,v) > 0$ then $d(u,v) - \bdist[1](u,v) = m$. Therefore $\supp(\omega) \subseteq M$.
\item Can be argued symmetrically to the previous case.
\end{enumerate}
Therefore, $M$ is a nonempty self-closed relation with respect to $d$.
\end{proof}

\begin{algorithm}[t]
\DontPrintSemicolon
Initialise $\C = (f,\rho)$ as an arbitrary vertex coupling structure for $\mathcal{A}$\;
$\textsc{isMin} \gets \mathit{false}$\;
\While{$\neg${\normalfont\textsc{isMin}}}{ \label{spi1:line3}
	\While{$\exists (s,t).\, \Delta_1(\gamma_1^\C)(s,t) < \gamma_1^\C(s,t)$}{  \label{spi1:line4}
		$R \gets h(\gamma_1^\C)(s,t)$\;
		$\C \gets \big(k(\gamma_1^\C), \, \rho[(s,t) / R] \big)$ \label{spi1:line6} \tcc*[r]{update coupling structure}
	} \label{spi1:line7}
	Let $M \gets {\approx_{\gamma_1^\C}}$ \tcc*[r]{note that $\gamma_1^\C = \Delta_1(\gamma_1^\C)$} \label{spi1:line8}
	\eIf{$M = \emptyset$}{ \label{spi1:line9}
		 {\normalfont\textsc{isMin}} $\gets$ \textit{true} \tcc*[r]{$\gamma_1^\C = \bdist[1]$} \label{spi1:line10}
	}{
		Compute $d = (\gamma_1^\C)_M$ as in Lemma~\ref{lem:decreasefixpoint}\; \label{spi1:line12}
		Re-initialise $\C$ as a vertex coupling structure s.t.\ $\Gamma_1^\C(d) = \Delta_1(d)$\; \label{spi1:line13}
	}
} \label{spi1:line15}
\Return $\gamma_1^\C$
\caption{Simple policy iteration algorithm computing $\bdist[1]$.}
\label{alg:bdist1}
\end{algorithm}

Algorithm~\ref{alg:bdist1} extends the procedure described in Section~\ref{sec:discountedcase} by encapsulating
the policy iteration update (lines~\ref{spi1:line4}--\ref{spi1:line7}) into an outer-loop
(lines~\ref{spi1:line3}--\ref{spi1:line15})
that is responsible to check whether the fixed point $\gamma_1^{\C_i}$ returned is the minimal one.
According to Lemma~\ref{lem:prefixfpoint}, $\Delta_1(\gamma^{\C_i}_1) \sqsubseteq \gamma^{\C_i}_1$. Hence, when we reach line~\ref{spi1:line8}, we have that $\Delta_1(\gamma^{\C_i}_1) = \gamma^{\C_i}_1$. Therefore, by Lemmas~\ref{lem:decreasefixpoint} and~\ref{lem:termination}, $\gamma^{\C_i}_1 = \bdist[1]$ if and only if $M = {\approx_{\gamma^{\C_i}_1}}$ is empty.
If $M$ is empty, we set the variable \textsc{isMin} to $\mathit{true}$ (line~\ref{spi1:line10}) causing the outer-loop to terminate. Otherwise, we construct $d = {(\gamma_1^{\C_i})}_M$ as in
Lemma~\ref{lem:decreasefixpoint} (line~\ref{spi1:line12})
and re-start the inner-loop from a vertex coupling structure
$\C_{i+1}$ such that $\Gamma_1^{\C_{i+1}}(d) = \Delta_1(d)$ (line~\ref{spi1:line13}) (\franck{\eg}, by using $\C_{i+1} = (k(d),\rho)$ where $\rho(s,t) = h(d)(s,t)$ for all $s,t \in S$).
As proven in Theorem~\ref{th:alg:bdist1}, $\gamma_1^{\C_{i}} \sqsupset \gamma_1^{\C_{i+1}}$. This guarantees a strict improvement of the discrepancy towards the minimal one.
Termination of Algorithm~\ref{alg:bdist1} is justified by similar arguments as for the discounted case.

\begin{thm}%
\label{th:alg:bdist1}
Algorithm~\ref{alg:bdist1} terminates and computes $\bdist[1]$.
\end{thm}
\begin{proof}
\change{First we prove termination. Recall that $\{\gamma_1^\C \mid \text{$\C$ vertex coupling structure for $\mathcal{A}$}\}$ is finite. }
Towards a contradiction, assume that Algorithm~\ref{alg:bdist1} does not terminate. Let $\C_0, \C_1, \C_2, \dots$ be the infinite sequence of coupling structures updates generated during a non-terminating execution of Algorithm~\ref{alg:bdist1}. Since the set above is finite, there must be $i < j$ such that $\gamma_1^{\C_i} = \gamma_1^{\C_j}$.
On the contrary, next we prove that the updates of the coupling structures in Algorithm~\ref{alg:bdist1} ensure that for all $n \in \mathbb{N}$, $\gamma_1^{\C_n} \sqsupset \gamma_1^{\C_{n+1}}$. Let $n \in \mathbb{N}$. We consider two cases:
\begin{itemize}
\item Assume $\Delta_1(\gamma_1^{\C_n})(s,t) < \gamma_1^{\C_n}(s,t)$, for some $s,t \in S$. In this case $\C_{n+1}$ is obtained from the update performed in line~\ref{spi1:line6}, which is exactly the one prescribed by Lemma~\ref{lem:update_improves}. Therefore we have that $\gamma_\lambda^{\C_n} \sqsupset \gamma_\lambda^{\C_{n+1}}$.

\item Assume $\Delta_1(\gamma_1^{\C_n})(s,t) \geq \gamma_1^{\C_n}(s,t)$ for all $s,t \in S$, i.e., $\gamma_1^{\C_n} \sqsubseteq \Delta_1(\gamma_1^{\C_n})$. By Lemma~\ref{lem:prefixfpoint} $\gamma_1^{\C_n} \sqsupseteq \Delta_1(\gamma_1^{\C_n})$, therefore $\gamma_1^{\C_n} = \Delta_1(\gamma_1^{\C_n})$. In this case $\C_{n+1}$ is constructed as a vertex coupling structure such that
$\Gamma_1^{\C_{n+1}}(d) = \Delta_1(d)$ where $M = {\approx_{\gamma_1^{\C_n}}} \neq \emptyset$ and $d = {(\gamma_1^{\C_n})}_M$ (see line~\ref{spi1:line13}). Then the following inequalities hold
\begin{align*}
	\gamma_1^{\C_n} &{} \sqsupset d \sqsupseteq \Delta_1(d) \tag{Lemma~\ref{lem:decreasefixpoint}} \\
	& = \Gamma_1^{\C_{n+1}}(d) \tag{by construction} \\
	& \sqsupseteq \gamma_1^{\C_{n+1}} \tag{by $\Gamma_1^{\C_{n+1}}(d) \sqsubseteq d$ and Knaster-Tarski fixed point theorem}
\end{align*}
\end{itemize}
This concludes the proof that $\gamma_1^{\C_n} \sqsupset \gamma_1^{\C_{n+1}}$ for all $n \in \mathbb{N}$.

When the execution of Algorithm~\ref{alg:bdist1} has reached line~\ref{spi1:line10} we have that $\gamma_1^\C = \Delta_1(\gamma_1^\C)$. Moreover, we have ${\approx_{\gamma_1^\C}} = \emptyset$. Therefore, by Lemma~\ref{lem:termination} we have that $\gamma_1^\C = \bdist[1]$. From here \textsc{isMin} is set to $\mathit{true}$. This prevents further executions of the body of the outer-loop. Therefore Algorithm~\ref{alg:bdist1} reached the return statement with $\gamma_1^\C = \bdist[1]$.
\end{proof}

\subsection{Experimental Results}
In this section, we evaluate the performance of the simple policy iteration algorithms on a
collection of randomly generated probabilistic automata. All the algorithms have been
implemented in Java and the source code is publicly available\footnote{
\url{https://bitbucket.org/discoveri/probabilistic-bisimilarity-distances-probabilistic-automata}.}.

The performance of Algorithm~\ref{alg:bdistlambda} has been compared with an implementation of the \emph{value iteration algorithm} proposed by Fu~\cite[Section 4]{Fu12}.
This algorithm works as follows. Starting from the bottom element, it iteratively applies $\Delta_\lambda$ to the current distance function generating the increasing chain $\mathbf{0} \sqsubseteq \Delta_\lambda(\mathbf{0}) \sqsubseteq  \Delta^2_\lambda(\mathbf{0}) \sqsubseteq \dots \sqsubseteq \Delta^{k-1}_\lambda(\mathbf{0}) \sqsubseteq \Delta^k_\lambda(\mathbf{0})$.



For each input instance, the comparison involves the following steps:
\begin{enumerate}
  \item\label{item:step1} We run Algorithm~\ref{alg:bdistlambda}, storing execution time, the
  number of solved transportation problems, and the number of coupling structures generated during the execution (i.e., the number of times a $\lambda$-discrepancy has been computed);
  \item\label{item:step2} Then, on the same instance, we execute the value iteration algorithm until the running
  time exceeds that of step~\ref{item:step1}. We report the execution time, the number of solved transportation problems, and the number of iterations.
\item Finally, we report the error $\max_{s,t \in S} |\bdist(s,t) - d(s,t)|$ between the distance $\bdist$ computed in step~\ref{item:step1} and the approximate result $d$ obtained in step~\ref{item:step2}.
\end{enumerate}
This has been done for a collection of automata varying from $10$ to $50$ states.
For each $n = 10, \dots, 50$, we considered $100$ randomly generated probabilistic automata, varying probabilistic out-degree and nondeterministic out-degree.
Table~\ref{tab:policy2} reports the average results of the comparison.
Our algorithm is able to compute the solution before value iteration can under-approximate it with an error ranging from $0.004$ to $0.06$ which is a non negligible error considering that we fixed $\lambda = 0.8$ and the distance has values in $[0,1]$.

%
\begin{table}[t]
	\centering
	\begin{tabular}{cccccccc}
    \toprule
	\multirow{2}{*}{$n = |S|$}&	\multicolumn{3}{c}{Simple Policy Iteration} & \multicolumn{3}{c}{Value Iteration}&	\multirow{2}{*}{Error} \\
	\cmidrule(lr){2-4}
	\cmidrule(lr){5-7}
	&	time (sec)&	\# TP&	\# $\C$&	time (sec)&	\# TP&	\# Iter&	\\
	\midrule
10 & 0.122   & 360.9   & 25.0  & 0.138   & 607.1   & 5.8  & 0.03018 \\
11 & 0.167   & 457.2   & 30.5  & 0.189   & 779.5   & 6.0  & 0.03090 \\
12 & 0.238   & 565.9   & 37.4  & 0.265   & 976.0   & 6.3  & 0.02814 \\
13 & 0.309   & 679.3   & 44.2  & 0.351   & 1177.7  & 6.4  & 0.04198 \\
14 & 0.412   & 813.4   & 52.2  & 0.463   & 1443.0  & 6.7  & 0.03673 \\
15 & 0.569   & 963.3   & 61.1  & 0.634   & 1764.0  & 6.9  & 0.03790 \\
20 & 2.694   & 1881.4  & 113.4 & 2.874   & 3325.5  & 7.6  & 0.03781 \\
30 & 15.832  & 4629.1  & 263.6 & 16.642  & 10309.5 & 9.9  & 0.02615 \\
40 & 99.433  & 17812.6 & 710.8 & 104.036 & 30870.3 & 8.8  & 0.01162 \\
50 & 137.624 & 13985.4 & 753.0 & 144.098 & 35118.3 & 12.3 & 0.00975 \\
\bottomrule
	\end{tabular}
	\caption{Comparison between 
	Simple Policy and Value Iteration Algorithm. %
	Average performance conducted on 100 randomly generated automata with number of states
	$n=10..50$,
	nondeterministic out-degree $k = 1..3$, and probabilistic out-degree $p = 2..3$. Discount $\lambda = 0.8$; accuracy $0.000001$.}\label{tab:policy2}
\end{table}

\begin{figure}[t]
\centering
\includegraphics[width=.95\textwidth]{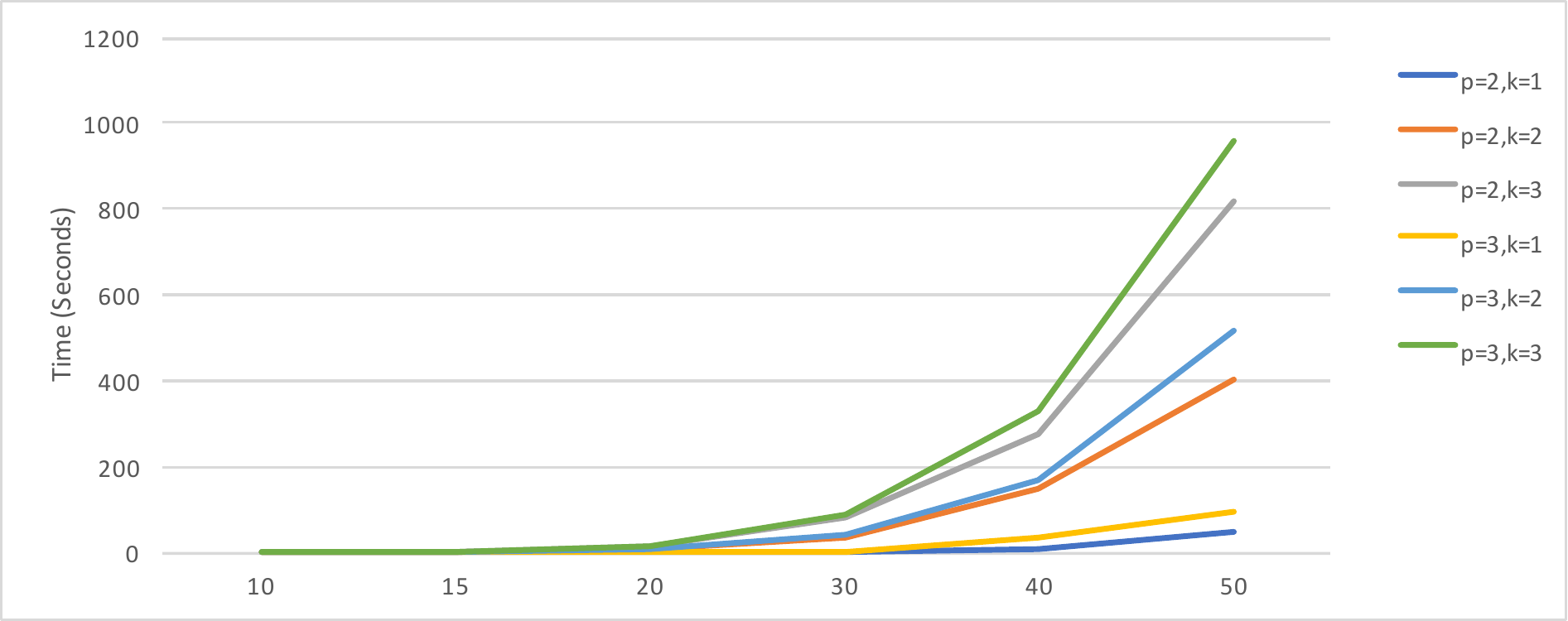}
\caption{Average performance for the Simple Policy Iteration Algorithm conducted on 100 randomly generated automata varying number of states $n=10..50$,  nondeterministic out-degree $k = 1..3$, and probabilistic out-degree $p = 2..3$. Discount factor $\lambda = 0.8$; accuracy $0.000001$.}%
\label{fig:graph}
\end{figure}
Furthermore in Figure~\ref{fig:graph} we observe that the execution time of the simple policy iteration algorithm
is particularly influenced by the degree of nondeterminism of the automaton. This may be explained by the fact that the current implementation uses a linear program for computing the $\lambda$-discrepancy (\cf\ Remark~\ref{rem:discrepancy}) which has $O(n^2  k^2)$ variables and $O(n^2  k^2)$ constraints where $n$ and $k$ are the number of states and the nondeterministic out-degree of the automaton, respectively.

Algorithm~\ref{alg:bdist1} extends the simple policy iteration algorithms proposed in~\cite{BacciLM:tacas13,TangB16} for labelled Markov chains. As pointed out in~\cite{T18}, implementations based on the decision procedure for the existential fragment of the first-order theory of the reals struggle to handle labelled Markov chains with a fifty states. For probabilistic automata, the algorithms in~\cite{ChatterjeeAMR08,ChatterjeeAMR10} suffer from the same problem.
%
%
The performance of Algorithm~\ref{alg:bdist1} is comparable to that of Algorithm~\ref{alg:bdistlambda} (\cf\ Table~\ref{tab:policy1}). Despite the fact that the simple policy algorithm is not guaranteed to be sound when the discount factor equals one, our experiments show that in practice a single iteration of the outer-loop of Algorithm~\ref{alg:bdist1} is often sufficient to yield the correct solution, \change{although it is still not clear to
us how to precisely characterise the conditions under which this situation happens.}

\begin{table}[t]
	\centering
	\begin{tabular}{ccccccccc}
		\toprule
		\multirow{2}{*}{$n = |S|$}&	\multicolumn{4}{c}{Simple Policy Iteration} & \multicolumn{3}{c}{Value Iteration}&	\multirow{2}{*}{Error} \\
        \cmidrule(lr){2-5}
        \cmidrule(lr){6-8}
		&	time (sec)&	\# TP&	\# $\C$&	\# Iter	&time (sec)&	\# TP&	\# Iter&	\\
		\midrule
10 & 0.129   & 394.8   & 25.9  & 1 & 0.144   & 647.4   & 6.2  & 0.07726 \\
11 & 0.179   & 513.3   & 32.1  & 1 & 0.202   & 837.1   & 6.4  & 0.08236 \\
12 & 0.263   & 655.1   & 39.1  & 1 & 0.293   & 1080.1  & 6.8  & 0.08988 \\
13 & 0.352   & 815.2   & 46.6  & 1 & 0.394   & 1310.5  & 7.0  & 0.10222 \\
14 & 0.465   & 966.6   & 53.9  & 1 & 0.514   & 1615.5  & 7.3  & 0.11156 \\
15 & 0.703   & 1159.5  & 62.5  & 1 & 0.786   & 2042.7  & 7.7  & 0.12045 \\
20 & 3.044   & 2291.7  & 111.8 & 1 & 3.316   & 3845.2  & 8.7  & 0.13823 \\
30 & 15.905  & 5088.6  & 223.1 & 1 & 16.929  & 9956.3  & 10.2 & 0.16030 \\
40 & 44.354  & 8597.3  & 364.0 & 1 & 47.580  & 16527.8 & 10.5 & 0.18830 \\
50 & 113.586 & 13484.7 & 545.0 & 1 & 121.639 & 28694.5 & 11.7 & 0.20939\\
	\bottomrule
	\end{tabular}
	\caption{Average performance of Algorithm~\ref{alg:bdist1} conducted on 100 randomly generated automata with number of states
	$n=10..50$, nondeterministic out-degree $k = 1..3$, and probabilistic out-degree $p = 2..3$. Discount $\lambda = 1$; accuracy $0.000001$.}\label{tab:policy1}
\end{table}

\change{
\begin{rem}[Worst-case Time Complexity]
Tang and van Breugel~\cite{TangB16,T18} showed that in the worst case, the simple policy iteration algorithm takes exponential time even with automata that are purely probabilistic (i.e., labelled Markov chains). The addition of nondeterminism also contributes in the exponential growth of the running time (\cf~Figure~\ref{fig:graph}).
\end{rem}
}


\section{Relation with Probabilistic Model Checking}%
\label{sec:linearProperties}

In this section we show how the probabilistic bisimilarity distance of Deng et al.\ relates to the problem of model checking $\omega$-regular specifications against probabilistic automata, where the nondeterministic choices are resolved by randomized schedulers.

Probabilistic automata are used for the verification of concurrent probabilistic systems, where the choice of how to interleave the executions of the parallel components is modelled by means of nondeterminism in the choice of the next transition to be taken.
Technically, an execution of a probabilistic automaton $\mathcal{A} = (S, L, \to, \ell)$ is an infinite sequence $s_0 s_1 \ldots \in S^\omega$ of labelled states obtained by taking a succession of probabilistic transitions $s_i \to \mu_{i}$ such that $\mu_i(s_{i+1}) > 0$, for each $i \in \mathbb{N}$.
The choice of the transition to be taken at each state is resolved by means of a scheduler. In this paper we consider \emph{randomized schedulers}, i.e., functions $\pi \colon S^+ \to \Distr{S}$ mapping a nonempty and finite sequence of states $s_0 \dots s_n \in S^+$ (the execution history) to a convex combination of distributions of the form $\sum_{s_n \to \mu} \alpha_{\mu} \cdot \mu \in \Distr{S}$, for some $\alpha_\mu \in [0,1]$ such that $\sum_{s_n \to \mu} \alpha_{\mu} =  1$. Roughly, a randomized scheduler decides the probability with which the next transition is chosen given the history of visited states.

The combination of a probabilistic automata $\mathcal{A}$ with a scheduler $\pi$ induces a Markov chain
on a random variable $X = (X_0, X_1, \dots) \in S^\omega$ on the measurable space of infinite sequences
over $S$ with distribution
\begin{equation*}
 \Pr[\pi]{s}{X_0= s_0, \dots, X_n = s_n} = \dirac_s(s_0) \cdot
 \prod_{i = 0}^{n-1} \pi(s_0 \dots s_{i})(s_{i+1}) \,,
\end{equation*}
where $\dirac_x$ denotes the Dirac distribution concentrated at $x$. The above describes the probability of executing the sequence of steps $s_0 \dots s_n$ by starting from the state $s$
under the randomized scheduler $\pi$.

\change{The measurable sets of $S^\omega$ are the elements of the infinite product
$\sigma$-algebra ${(2^S)}^\omega$, i.e., the smallest $\sigma$-algebra containing the subsets of the form $s_0 \dots s_n S^\omega = \{s_0 \dots s_n w \mid w \in S^\omega \}$ (a.k.a., \emph{discrete cylinders}), for arbitrary $n \in \mathbb{N}$, $s_i \in S$ and $0 \leq i \leq n$. Measurable sets are the subsets of sequences where the probability measure $\Pr[\pi]{s}{}$ is well-defined.
For a measurable set $H \subseteq S^\omega$, we denote by $\Pr[\pi]{s}{X \in H}$ the probability that an execution starting from $s$ under the scheduler $\pi$ belongs to $H$.

\bigskip

Rather than measuring the probability of concrete executions over $S^\omega$, one is often more interested in the probability that certain execution traces satisfy abstract properties
over the measurable space $L^\omega$ of infinite sequences of labels, representing the sequences of atomic properties satisfied by concrete executions of the automaton.

Formally, for a measurable set $E \subseteq L^\omega$, we denote by $\Pr[\pi]{s}{\ell(X) \in E}$ the probability that an execution generates a sequence of labels in $E$, where $\ell(X) = (\ell(X_0), \ell(X_1), \dots) \in L^\omega$ is the random variable induced from $X$ by the labelling function $\ell$.
The $\sigma$-algebra of $L^\omega$ contains all the $\omega$-regular languages expressible over the alphabet $L$~\cite[Chapter~10]{BaierK08}. This means that the probability of the runs of $\mathcal{A}$ of satisfying $\omega$-regular properties, possibly expressed in the form of LTL formulas, can be formally measured by $\Pr[\pi]{s}{}$, hence allowing the quantitative analysis of probabilistic automata.}

The quantitative analysis of a probabilistic automaton $\mathcal{A}$ against $\omega$-regular specifications, more commonly known as \emph{probabilistic model-checking}, amounts to establishing the maximal and minimal probability of satisfying $\omega$-regular properties
$E \subseteq L^\omega$ over infinite sequences of labels from a starting state $s$.
Formally, this corresponds to computing
\begin{equation*}
\textstyle
  \mathrm{Max}_s(E) = \sup_{\pi \in \Pi} \Pr[\pi]{s}{\ell(X) \in E}
  \qquad \text{and} \qquad
  \mathrm{Min}_s(E) = \inf_{\pi \in \Pi} \Pr[\pi]{s}{\ell(X) \in E}
\end{equation*}
where the infimum and supremum are taken over the set $\Pi$ of all randomized schedulers. Note that, considering minimal or maximal probabilities corresponds to a worst/best-case analysis (see~\cite[Chapter~10]{BaierK08} for more details).

The following is the main result of the section. It states that the probabilistic bisimilarity distance
bounds the difference between maximal and minimal probability of satisfying any measurable linear-time property (\franck{\eg}, $\omega$-regular specifications) on two given initial states.
\begin{thm}%
\label{th:minmaxbound}
For all measurable subsets $E \subseteq L^\omega$,
\begin{equation*}
| \mathrm{Max}_s(E) - \mathrm{Max}_t(E) | \leq \bdist[1](s,t)
\qquad \text{and} \qquad
| \mathrm{Min}_s(E) - \mathrm{Min}_t(E) |  \leq \bdist[1](s,t) \,.
\end{equation*}
\end{thm}
The above can be seen as a quantitative generalization of the folklore result that bisimilar states satisfy the
same linear-time properties with the same probability.

\begin{rem}%
\label{rem:approxModel-checking}
The relevance of Theorem~\ref{th:minmaxbound} is not just theoretical, but could possibly lead to significant practical applications. Imagine that the distance $\bdist[1](s,t)$ between some given states $s$ and $t$ is small (and known); then, computing $\mathrm{Max}_s(E)$ (resp.\ $\mathrm{Min}_s(E)$) in the state $s$ may be enough for obtaining a good approximation for the actual value of $\mathrm{Max}_t(E)$ (resp.\ $\mathrm{Min}_t(E)$) without the need of computing it on the state $t$. This approach may lead to savings in the overall model-checking time of $t$, especially if the executions starting from $s$ have a significant reduced degree of nondeterminism than whose starting from $t$.
\end{rem}

The proof of Theorem~\ref{th:minmaxbound} is based on the coupling characterisation of the bisimilarity distance
presented in Theorem~\ref{th:mincoupling} and the following technical lemma (Lemma~\ref{lem:variationaldiscrepancy}) which establishes under which conditions the discrepancy associated with a coupling structure can be used to bound the variational distance between \change{the probability induced by a probabilistic automaton $\mathcal{A}$ under two different schedulers.
Specifically, we establish how, from a coupling structure $\C$, one can retrieve
a set-coupling $R_\C \in  \mathcal{R}(\Pi, \Pi)$ of schedulers for $\mathcal{A}$ such that for each pair of schedulers $(\pi, \pi')$ for $\mathcal{A}$ in $R_\C$, the difference $|\Pr[\pi]{s}{\ell(X) \in E} - \Pr[\pi']{t}{\ell(X) \in E}|$ for any measurable $E \subseteq L^\omega$, is bounded by the discrepancy
$\gamma^{\C}_1(s,t)$.

The definition of $R_\C$ can be intuitively understood by recalling that $\gamma^{\C}_1(s,t)$ corresponds to the maximal probability of reaching a pair of states with different labels from the state pair $(s,t)$ in the automaton $\mathcal{A}_\C$ induced from $\C$. Roughly, $R_\C$ collects all the pairs of schedulers for $\mathcal{A}$ obtained as the
left and right projection of a scheduler for $\mathcal{A}_\C$. How the projections are defined is technical and the interested reader can found the formal definition in the proof. }

\begin{lem}%
\label{lem:variationaldiscrepancy}
\change{For any coupling structure $\C$ for $\mathcal{A}$ and $s,t \in S$, exists
$R_\C \in \mathcal{R}(\Pi, \Pi)$ such that, for all measurable $E \subseteq L^\omega$ and $(\pi, \pi') \in R_\C$,}
$ |\Pr[\pi]{s}{\ell(X) \in E} - \Pr[\pi']{t}{\ell(X) \in E}| \leq \gamma^{\C}_1(s,t)$.
\end{lem}
\begin{proof}
Fix $s, t \in S$ and $\C  = (f, \rho)$ a coupling structure for $\mathcal{A}$.
Let $\mathcal{A}_\C$ be the automaton associated with the coupling structure $\C$.

We split the proof in two parts. (Part 1) deals with the definition of the set-coupling
$R_\C \in \mathcal{R}(\Pi, \Pi)$; (Part 2) with proving that
$ |\Pr[\pi]{s}{\ell(X) \in E} - \Pr[\pi']{t}{\ell(X) \in E}| \leq \gamma^{\C}_1(s,t)$, for all pairs of
randomized schedulers $(\pi, \pi') \in R_\C$.

Hereafter, for a nonempty finite sequence $\sigma \in S^{+}$ and a random variable $X$ on
$S^\omega$, we use $X \prec \sigma$ to denote $X \in \sigma S^\omega$.

\begin{description}
\item[Part 1]
\newcommand{\zip}[2]{\langle#1,#2\rangle}
Let $(X,Y) \in S^\omega \times S^\omega$ be the random variable describing
the infinite sequence of state pairs along a run of $\mathcal{A}_\C$. Then, for any two nonempty finite sequences of the same length $\sigma = \sigma_0 \dots \sigma_n$ and $\tau = \tau_0 \dots \tau_n$ over $S$,
\begin{equation*}
 \Pr[\pi]{(s,t)}{(X,Y) \prec \zip{\sigma}{\tau}} = \Pr[\pi]{(s,t)}{X \prec \sigma, Y \prec \tau}
\end{equation*}
is the probability that, starting from $(s,t)$, a run of $\mathcal{A}_\C$ under the scheduler $\pi$ has prefix
$\zip{\sigma}{\tau} = (\sigma_0,\tau_0) \dots (\sigma_{n},\tau_{n})$. The above can be alternatively
formulated in terms of conditional probabilities in the following two ways:
\begin{align}
\Pr[\pi]{(s,t)}{X \prec \sigma, Y \prec \tau}
&=
\Pr[\pi]{(s,t)}{X \prec \sigma} \cdot
\Pr[\pi]{(s,t)}{Y \prec \tau \mid X \prec \sigma } \,,
\label{eq:conditionalprob}
\\
\Pr[\pi]{(s,t)}{X \prec \sigma, Y \prec \tau}
&=
\Pr[\pi]{(s,t)}{Y \prec \tau} \cdot
\Pr[\pi]{(s,t)}{X \prec \sigma \mid Y \prec \tau} \,.
\end{align}

\medskip
Given a scheduler $\pi$ for $\mathcal{A}_\C$, we define
the maps $\pi_S, \pi_T \colon S^+ \to \Distr{S}$ as follows%
\footnote{We assume that $\Pr[\pi]{(s,t)}{ E \mid F } = 0$, whenever $\Pr[\pi]{(s,t)}{ F } = 0$ for any two measurable events $E, F$.}, for arbitrary nonempty sequences
$\sigma, \tau$ over $S$
\begin{align}
\pi_S(\sigma)(u)
&= \sum_{\tau \in S^{|\sigma|}}
\Big( \Pr[\pi]{(s,t)}{ Y \prec \tau \mid X \prec \sigma }
\sum_{v \in S} \pi(\zip{\sigma}{\tau})(u,v) \Big)  \,,
\label{eq:piS}
\\
\pi_T(\tau)(u)
&= \sum_{\sigma \in S^{|\tau|}}
\Big( \Pr[\pi]{(s,t)}{X \prec \sigma \mid Y \prec \tau }
\sum_{v \in S} \pi(\zip{\sigma}{\tau})(u,v) \Big) \,.
\label{eq:piT}
\end{align}
\change{We call $\pi_S$ and $\pi_T$, the left and right projections of $\pi$.}

Intuitively, $\pi_S(\sigma)(u)$ describes the probability that under the scheduler $\pi$
a run of $\mathcal{A}_\C$ with initial state $(s,t)$ has a prefix of the form $\zip{\sigma u}{\tau'}$, for some $\tau' \in S^{|\sigma|+1}$; symmetrically, $\pi_T(\tau)(u)$ is the probability that the prefix is of the form $\zip{\sigma'}{\tau u}$, for some $\sigma' \in S^{|\tau|+1}$.

Next, we prove that $\pi_S$ and $\pi_T$ are well-defined schedulers for $\mathcal{A}$.
 We provide the proof only for $\pi_S$, as it is similar for $\pi_T$.
We need to show that
$\pi_S(\sigma)$ is a convex combination of the form $\sum_{\mu \in \delta(\sigma_{n})} \alpha_{\mu} \cdot \mu$, for some $\alpha_\mu \in [0,1]$ such that $\sum_{\mu \in \delta(\sigma_{n})} \alpha_{\mu} =  1$.
By hypothesis that $\pi$ is a scheduler for $\mathcal{A}_\C$, we have that
$\pi(\zip{\sigma}{\tau}) = \sum_{(\mu,\nu) \in \rho(\sigma_{n},\tau_{n})} \xi_{\mu,\nu} \cdot f(\mu,\nu)$,
for some $\xi_{\mu,\nu} \in [0,1]$, such that $\sum_{(\mu,\nu) \in \rho(\sigma_n,\tau_n)} \xi_{\mu,\nu} =  1$. Without loss of generality, we can assume that $\xi_{\mu,\nu} = 0$, whenever
$(\mu,\nu) \notin \rho(\sigma_n,\tau_n)$.
Let $\kappa_{\sigma,\tau} = \Pr[\pi]{(s,t)}{ Y \prec \tau \mid X \prec \sigma }$, then
\begin{align*}
\pi_S(\sigma)(u)
&= \sum_{\tau \in S^{|\sigma|}}
\Big( \kappa_{\sigma,\tau} \sum_{v \in S} \pi(\zip{\sigma}{\tau})(u,v) \Big)
\tag{eq.~\eqref{eq:piS}} \\
&= \sum_{\tau \in S^{|\sigma|}}
\Big( \kappa_{\sigma,\tau} \sum_{v \in S}
\sum_{(\mu,\nu) \in \rho(\sigma_n,\tau_n)} \xi_{\mu,\nu} \cdot f(\mu,\nu)(u,v) \Big)
\tag{$\pi$ scheduler} \\
&= \sum_{\tau \in S^{|\sigma|}}
\Big( \kappa_{\sigma,\tau} \sum_{(\mu,\nu) \in \rho(\sigma_n,\tau_n)} \xi_{\mu,\nu} \sum_{v \in S} f(\mu,\nu)(u,v) \Big)  \\
&= \sum_{\tau \in S^{|\sigma|}}
\Big( \kappa_{\sigma,\tau} \sum_{(\mu,\nu) \in \rho(\sigma_n,\tau_n)} \xi_{\mu,\nu} \cdot \mu(u) \Big)
\tag{$f$ measure-coupling map} \\
&= \sum_{\tau \in S^{|\sigma|}}
\Big( \kappa_{\sigma,\tau} \sum_{\mu \in \delta(\sigma_n)} \sum_{\nu \in \delta(\tau_n)} \xi_{\mu,\nu} \cdot \mu(u) \Big)
\tag{$\rho$ set-coupling map} \\
&= \sum_{\mu \in \delta(\sigma_n)} \Big(
\sum_{\tau \in S^{|\sigma|}} \kappa_{\sigma,\tau} \sum_{\nu \in \delta(\tau_n)} \xi_{\mu,\nu}
\Big) \cdot \mu(u)
\notag
\end{align*}
By letting $\alpha_\mu = \sum_{\tau \in S^{|\sigma|}} \kappa_{\sigma,\tau} \sum_{\nu \in \delta(\tau_n)} \xi_{\mu,\nu}$, we get $\pi_S(\sigma) = \sum_{\mu \in \delta(\sigma_n)} \alpha_{\mu} \cdot \mu$ in the desired form.
Next we show that this is a convex combination, i.e., $\sum_{\mu \in \delta(\sigma_n)} \alpha_{\mu} =  1$.
\begin{align*}
\sum_{\mu \in \delta(\sigma_n)} \alpha_{\mu}
&= \sum_{\mu \in \delta(\sigma_n)} \sum_{\tau \in S^{|\sigma|}} \kappa_{\sigma,\tau} \sum_{\nu \in \delta(\tau_n)} \xi_{\mu,\nu}
\tag{def.\ $\alpha_\mu$} \\
&= \sum_{\tau \in S^{|\sigma|}} \kappa_{\sigma,\tau} \sum_{\mu \in \delta(\sigma_n)} \sum_{\nu \in \delta(\tau_n)} \xi_{\mu,\nu}
\notag  \\
&= \sum_{\tau \in S^{|\sigma|}} \kappa_{\sigma,\tau} \sum_{(\mu,\nu) \in \rho(\sigma_n,\tau_n)} \xi_{\mu,\nu}
\tag{$\rho$ set-coupling map}  \\
&= \sum_{\tau \in S^{|\sigma|}} \kappa_{\sigma,\tau}
\tag{$\sum_{(\mu,\nu) \in \rho(\sigma_n,\tau_n)} \xi_{\mu,\nu} =  1$}  \\
&= \sum_{\tau \in S^{|\sigma|}} \Pr[\pi]{(s,t)}{ Y \prec \tau \mid X \prec \sigma }
\tag{def.\ $\kappa_{\sigma,\tau}$}  \\
&= 1
\tag{probability}
\end{align*}

So, $\pi_S$ and $\pi_T$ are well-defined schedulers for $\mathcal{A}$.
Given the above, we define the relation $R_\C \subseteq \Pi \times \Pi$ on schedulers for
$\mathcal{A}$ by
\begin{equation*}
  R_\C = \{  (\pi_S, \pi_T) \mid  \text{$\pi$ scheduler on $\mathcal{A}_\C$}  \} \,.
\end{equation*}
\change{To better understand the definition of $R_\C$, recall that $\mathcal{A}_\C$ can be interpreted
as the automaton describing the concurrent execution of two copies of $\mathcal{A}$ synchronised according to $\C$. Then, $\pi_S$ and $\pi_T$ can be interpreted as the schedulers on obtained from $\pi$, by respectively taking the left and right projections of the executions of $\mathcal{A}_\C$ as computations on $\mathcal{A}$. The relation $R_\C$ is given as the collection of these pair of projections.

Next we prove that $R_\C$ is a set-coupling for $(\Pi,\Pi)$, that is}
\begin{align*}
\{ \pi_1 \mid \exists \pi_2 \in \Pi.\, (\pi_1,\pi_2) \in R_\C \} = \Pi
&&\text{and} &&
\{ \pi_2 \mid \exists \pi_1 \in \Pi.\, (\pi_1,\pi_2) \in R_\C \} = \Pi \,.
\end{align*}
By definition of $R_\C$, this is equivalent to prove that for an arbitrary
pair of schedulers $\pi_S, \pi_T$ for $\mathcal{A}$ we can find a scheduler $\pi$ for $\mathcal{A}_\C$ such
that~\eqref{eq:piS},~\eqref{eq:piT} hold (hence, $(\pi_S, \pi_T) \in R_\C$).

Let $\sigma = \sigma_0 \dots \sigma_n$ and $\tau = \tau_0 \dots \tau_n$ be a pair of nonempty finite sequences of the same length over $S$, and assume
$\pi_S(\sigma) = \sum_{\mu \in \delta(\sigma_n)} \alpha^\sigma_{\mu} \cdot \mu$ and
$\pi_T(\tau) = \sum_{\nu \in \delta(\tau_n)} \beta^\tau_{\nu} \cdot \nu$, for some
$\alpha^\sigma_{\mu},\beta^\tau_{\nu} \in [0,1]$ such that
$\sum_{\mu \in \delta(\sigma_n)} \alpha^\sigma_{\mu} = 1$ and $\sum_{\nu \in \delta(\tau_n)} \beta^\tau_{\nu} = 1$. We define
\begin{align*}
\pi(\zip{\sigma}{\tau}) = \sum_{(\mu,\nu) \in \rho(\sigma_n,\tau_n)} \xi^{\sigma,\tau}_{\mu,\nu}
\cdot f(\mu,\nu) \,
&&
\text{where } \,
\xi^{\sigma,\tau}_{\mu,\nu} =
\begin{cases}
\alpha^\sigma_{\mu} \cdot \beta^\tau_{\nu} & \text{if $(\mu,\nu) \in \rho(\sigma_n,\tau_n)$} \\
0 & \text{otherwise} \,.
\end{cases}
\end{align*}

By the fact that $\rho(\sigma_n,\tau_n)$ is a set-coupling in $\mathcal{R}(\delta(\sigma_n), \delta(\tau_n))$ and the definition of $\xi^{\sigma,\tau}_{\mu,\nu}$, it is easy to see that for all
$\mu \in \delta(\sigma_n)$ and $\nu \in \delta(\tau_n)$
\begin{equation}
\sum_{\nu \in \delta(\tau_n)} \xi^{\sigma,\tau}_{\mu,\nu} = \alpha^\sigma_{\mu}
\qquad \text{and} \qquad
\sum_{\mu \in \delta(\sigma_n)} \xi^{\sigma,\tau}_{\mu,\nu} = \beta^\tau_{\nu} \,.
\label{eq:setcoupcondition}
\end{equation}
Next we show that~\eqref{eq:piS} holds.
Let $\kappa_{\sigma,\tau} = \Pr[\pi]{(s,t)}{ Y \prec \tau \mid X \prec \sigma }$, then
\begin{align*}
&\pi_S(\sigma)(u)   \\
&= \sum_{\tau \in S^{|\sigma|}}
\kappa_{\sigma,\tau}  \, \pi_S(\sigma)(u)
\tag{convex combination} \\
&= \sum_{\tau \in S^{|\sigma|}}
\kappa_{\sigma,\tau}  \Big(  \sum_{\mu \in \delta(\sigma_n)} \alpha^\sigma_{\mu} \cdot \mu(u)  \Big)
\tag{eq.~\eqref{eq:piS}} \\
&= \sum_{\tau \in S^{|\sigma|}}
\kappa_{\sigma,\tau}  \Big(
\sum_{\mu \in \delta(\sigma_n)} \sum_{\nu \in \delta(\tau_n)} \xi^{\sigma,\tau}_{\mu,\nu} \cdot \mu(u)  \Big)
\tag{by~\eqref{eq:setcoupcondition}} \\
&= \sum_{\tau \in S^{|\sigma|}}
\kappa_{\sigma,\tau}  \Big(
\sum_{\mu \in \delta(\sigma_n)} \sum_{\nu \in \delta(\tau_n)} \xi^{\sigma,\tau}_{\mu,\nu}
 \sum_{v \in S} f(\mu,\nu)(u,v)  \Big)
\tag{$f(\mu,\nu) \in \Omega(\mu,\nu)$} \\
&= \sum_{\tau \in S^{|\sigma|}}
\kappa_{\sigma,\tau}  \Big(
\sum_{(\mu,\nu) \in \rho(\sigma_n,\tau_n)} \xi^{\sigma,\tau}_{\mu,\nu}
 \sum_{v \in S} f(\mu,\nu)(u,v)  \Big)
\tag{$\rho$ set-coupl.\ map} \\
&= \sum_{\tau \in S^{|\sigma|}}
\Big( \kappa_{\sigma,\tau}   \sum_{v \in S}
\sum_{(\mu,\nu) \in \rho(\sigma_n,\tau_n)} \xi^{\sigma,\tau}_{\mu,\nu} \cdot
 f(\mu,\nu)(u,v)  \Big)
\notag \\
&= \sum_{\tau \in S^{|\sigma|}}
\Big( \Pr[\pi]{(s,t)}{ Y \prec \tau \mid X \prec \sigma }
\sum_{v \in S} \pi(\zip{\sigma}{\tau})(u,v) \Big)  \,
\tag{def.\ $\kappa_{\sigma,\tau}$ and $\pi$}
\end{align*}
Equation~\eqref{eq:piT} is proven symmetrically.

\item[Part 2]
We prove $\Pr[\pi]{(s,t)}{} \in \Omega(\Pr[\pi_S]{s}{}, \Pr[\pi_T]{t}{})$ first.
Showing the marginal conditions corresponds to prove that, for all nonempty sequences
$\sigma = \sigma_0 \dots \sigma_n$ and $\tau = \tau_0 \dots \tau_n$ over $S$,
\begin{align*}
\Pr[\pi_S]{s}{X \prec \sigma} = \Pr[\pi]{(s,t)}{X \prec \sigma}
&& \text{and} &&
\Pr[\pi_T]{t}{Y \prec \tau} = \Pr[\pi]{(s,t)}{Y \prec \tau} \,.
\end{align*}
We prove only the equality on the left, as the other is similar.
We proceed by induction on $n \geq 0$.
\begin{itemize}
\item Base case, $n = 0$. Then $\sigma = \sigma_0 \in S$ and
\begin{align*}
\Pr[\pi_S]{s}{X \prec \sigma_0}
&= \dirac_s(\sigma_0)
\tag{def.\ $\Pr[\pi_S]{s}{}$} \\
&= \sum_{\tau_0 \in S} \dirac_s(\sigma_0) \cdot \dirac_t(\tau_0)
\tag{convex combination} \\
&= \sum_{\tau_0 \in S} \dirac_{(s,t)}(\sigma_0,\tau_0)
\tag{def.\ $\dirac_{(s,t)}$} \\
&=\sum_{\tau_0 \in S} \Pr[\pi]{(s,t)}{X \prec \sigma_0, Y \prec \tau_0}
\tag{def.\ $\Pr[\pi]{(s,t)}{}$} \\
&= \Pr[\pi]{(s,t)}{X \prec \sigma_0}
\tag{additivity} \,.
\end{align*}
\item Inductive step, $n \geq 0$. Let $\sigma = \sigma_0 \dots \sigma_n$ and $s' \in S$, then

\begin{align*}
&\Pr[\pi_S]{s}{X \prec \sigma s'} \\
&= \Pr[\pi_S]{s}{X \prec \sigma} \cdot \pi_S(\sigma)(s')
\tag{def.\ $\Pr[\pi_S]{s}{}$} \\
&= \Pr[\pi]{(s,t)}{X \prec \sigma} \cdot \pi_S(\sigma)(s')
\tag{inductive hp} \\
&= \Pr[\pi]{(s,t)}{X \prec \sigma} \cdot
\sum_{\tau \in S^{|\sigma|}} \Big(\Pr[\pi]{(s,t)}{Y \prec \tau \mid X \prec \sigma } \cdot
\sum_{t' \in S} \pi(\zip{\sigma}{\tau})(s',t') \Big)
\tag{eq.~\eqref{eq:piS}} \\
&= \sum_{\tau \in S^{|\sigma|}}\sum_{t' \in S} \Pr[\pi]{(s,t)}{X \prec \sigma} \cdot
\Pr[\pi]{(s,t)}{Y \prec \tau \mid X \prec \sigma } \cdot \pi(\zip{\sigma}{\tau})(s',t')
\\
&= \sum_{\tau \in S^{|\sigma|}} \sum_{t' \in S}
\Pr[\pi]{(s,t)}{X \prec \sigma, Y \prec \tau} \cdot \pi(\zip{\sigma}{\tau})(s',t')
\tag{by~\eqref{eq:conditionalprob}} \\
&= \sum_{\tau \in S^{|\sigma|}} \sum_{t' \in S}
\Pr[\pi]{(s,t)}{X \prec \sigma s', Y \prec \tau t'}
\tag{def.\ $\Pr[\pi]{(s,t)}{}$} \\
&=\Pr[\pi]{(s,t)}{X \prec \sigma s'} \,.
\tag{additivity}
\end{align*}
\end{itemize}
The right-marginal condition is proven symmetrically.

Note that the discrepancy $\gamma_1^\C(s,t)$ is the maximal
probability of reaching a state pair $(u,v)$ such that $\ell(u) \neq \ell(v)$
by starting from the state pair $(s,t)$ in $\mathcal{A}_\C$. That is,
\begin{equation}
\gamma_1^\C(s,t) = \sup_\pi \Pr[\pi]{(s,t)}{\ell(X) \neq \ell(Y)} \,,
\label{eq:maxreach}
\end{equation}
where $\pi$ ranges over all schedulers for $\mathcal{A}_\C$.
Thus, from $\Pr[\pi]{(s,t)}{} \in \Omega(\Pr[\pi_S]{s}{}, \Pr[\pi_T]{t}{})$ we have
\begin{align*}
\Pr[\pi_S]{s}{E} &= \Pr[\pi]{(s,t)}{\ell(X) \in E} \\
&\geq \Pr[\pi]{(s,t)}{\ell(X) = \ell(Y), \ell(Y) \in E} \\
& = 1 - \Pr[\pi]{(s,t)}{\ell(X) \neq \ell(Y) \cup \ell(Y) \not\in E} \\
&\geq 1 -  \Pr[\pi]{(s,t)}{\ell(X) \neq \ell(Y)} -  \Pr[\pi]{(s,t)}{\ell(Y) \not\in E} \\
&= \Pr[\pi]{(s,t)}{\ell(Y) \in E} - \Pr[\pi]{(s,t)}{\ell(X) \neq \ell(Y)} \\
&\geq \Pr[\pi]{(s,t)}{\ell(Y) \in E} - \gamma_1^\C(s,t) \\
&= \Pr[\pi_T]{(s,t)}{E} - \gamma_1^\C(s,t) \,.
\end{align*}
From the above, we conclude that $ |\Pr[\pi_S]{s}{\ell(X) \in E} - \Pr[\pi_T]{t}{\ell(Y) \in E}| \leq \gamma^{\C}_1(s,t)$.
\end{description}
The proofs follows immediately by combining Part 1 and 2.
\end{proof}


Given Lemma~\ref{lem:variationaldiscrepancy} it is easy to establish the result stated in Theorem~\ref{th:minmaxbound}.
\begin{proof}[Proof of Theorem~\ref{th:minmaxbound}]
Let $d_{\mathbb{R}}(x,y) = | x - y |$ denote the Euclidean distance on the real line and define
$K = \{ \Pr[\pi]{s}{\ell(X) \in E} \mid \pi \in \Pi \}$ and $H = \{ \Pr[\pi]{t}{\ell(X) \in E} \mid \pi \in \Pi \}$.

Then, by~\cite[Lemma~3.2]{Memoli11}
\begin{align*}
\H{d_{\mathbb{R}}}(K,H)
&\geq \max \{  | \sup K - \sup H |, | \inf K - \inf H |  \} \\
&= \max \{  | \mathrm{Max}_s(E) - \mathrm{Max}_t(E) |, | \mathrm{Min}_s(E) - \mathrm{Min}_t(E) |  \} \,.
\end{align*}
Next we show that $\bdist[1](s,t) \geq \H{d_{\mathbb{R}}}(K,H)$.
\begin{align*}
&\H{d_{\mathbb{R}}}(K,H) \\
&= \inf \left\{ \sup_{(\pi,\pi') \in R} |\Pr[\pi]{s}{\ell(X) \in E} - \Pr[\pi']{t}{\ell(X) \in E}|  \Big| R \in \mathcal{R}(\Pi,\Pi) \right\}
\tag{Theorem~\ref{th:hausdorff}} \\
&\leq \inf \{ \gamma_1^\C  \mid \text{ $\C$ coupling structure for $\mathcal{A}$ }  \}
\tag{Lemma~\ref{lem:variationaldiscrepancy}} \\
&= \bdist[1](s,t) \,.
\tag{Theorem~\ref{th:mincoupling}}
\end{align*}
Therefore, $| \mathrm{Max}_s(E) - \mathrm{Max}_t(E) | \leq \bdist[1](s,t)$ and
$| \mathrm{Min}_s(E) - \mathrm{Min}_t(E) |  \leq \bdist[1](s,t)$.
\end{proof}

Another consequence of Lemma~\ref{lem:variationaldiscrepancy} is that the bisimilarity distance provides an upper bound of the Hausdorff lifting of the variational distance between sets of distributions induced by the Markov chains obtained by ranging over all possible schedulers. In the theorem we use $\tv{}$ to denote the
\emph{total variation distance} between probability measures, defined as $\tv{\mu, \nu} = \sup_E |\mu(E) - \nu(E)|$, where $E$ ranges over all measurable subsets.
\begin{thm}%
\label{th:tvbound}
$\H{\tv{}}(\{ \Pr[\pi]{s}{ \ell(X) \in \cdot } \mid \pi \in \Pi \}, \{ \Pr[\pi]{t}{ \ell(X) \in \cdot } \mid \pi \in \Pi \}) \leq \bdist[1](s,t)$.
\end{thm}
\begin{proof}
\begin{align*}
\H{\tv{}}&(\{ \Pr[\pi]{s}{ \ell(X) \in \cdot } \mid \pi \in \Pi \}, \{ \Pr[\pi]{t}{ \ell(X) \in \cdot } \mid \pi \in \Pi \}) = \\
&= \inf \left\{ \sup_{(\pi,\pi') \in R} \tv{\Pr[\pi]{s}{ \ell(X) \in \cdot }, \Pr[\pi']{t}{ \ell(X) \in \cdot }}   \mid R \in \mathcal{R}(\Pi,\Pi) \right\}
\tag{Theorem~\ref{th:hausdorff}} \\
&= \inf \left\{ \sup_{(\pi,\pi') \in R}    \sup_E | \Pr[\pi]{s}{\ell(X) \in E} - \Pr[\pi']{t}{\ell(X) \in E} |  \Big| R \in \mathcal{R}(\Pi,\Pi) \right\}
\tag{def.\ $\tv{}$} \\
&\leq \inf \{ \gamma_1^\C  \mid \text{ $\C$ coupling structure for $\mathcal{A}$ }  \}
\tag{Lemma~\ref{lem:variationaldiscrepancy}} \\
&= \bdist[1](s,t) \,.
\tag*{(Theorem~\ref{th:mincoupling}) \qedhere}
\end{align*}
\end{proof}

Theorem~\ref{th:tvbound} can be alternatively stated as follows. For any scheduler $\pi$ there exists a scheduler $\pi'$ such that $| \Pr[\pi]{s}{ \ell(X) \in E} -  \Pr[\pi']{t}{\ell(X) \in E} | \leq \bdist[1](s,t)$, for all measurable subsets $E \subseteq L^\omega$.



\section{Conclusion and Future Work}%
\label{sec:conclusion}

We presented a novel characterization of the probabilistic bisimilarity distance of Deng et al.~\cite{DengCPP06} as the solution of a simple stochastic game.
Starting from it, we designed algorithms for computing the distances based on Condon's simple policy iteration algorithm. The correctness of Condon's approach relies on the assumption that the input game is stopping. This may not be the case for our probabilistic bisimilarity games when the discount factor is one. We overcame this problem by means of an improved termination condition based on the notion of self-closed relation due to Fu~\cite{Fu12}.

As in~\cite{TangB16}, our simple policy iteration algorithm has exponential worst-case time complexity.
Nevertheless, experiments show that our method can compete in practice with the value iteration algorithm by
Fu~\cite{Fu12} which has theoretical polynomial-time complexity for $\lambda < 1$. To the best of our knowledge, our algorithm is the first practical solution for computing the bisimilarity distance when $\lambda = 1$, performing orders of magnitude faster than the existing solutions based on the existential fragment of the first-order theory of the reals~\cite{ChatterjeeAMR08,ChatterjeeAMR10,ChenHL07}.

As future work, we plan to improve on the current implementation in the line of~\cite{TangB18}, by exploiting the fact that bisimilar states and probabilistic distance one~\cite{TangB18:concur} can be efficiently pre-computed before starting the policy iteration.
We believe that this would yield a significant cut down in the time required to compute the discrepancy at each iteration which turned out to be the bottleneck of our algorithms.

More efficient algorithms might lead to the speedup of verification tools
for concurrent probabilistic systems, as behavioral distances relate to the satisfiability of logical properties.
For the case of labelled Markov chains, in~\cite{ChenBW12,BacciBLM15} the variational difference between two states with respect to their probability of satisfying linear-time properties (eg., LTL formulas) is shown to be bound by
the (undiscounted) probabilistic bisimilarity distance. In Section~\ref{sec:linearProperties} we showed that a similar result holds for the case of probabilistic automata with additional subtleties that arise by the need of handling the nondeterminism. In light of this relation it would be interesting to develop approximated techniques to cut down the overall model-checking time of probabilistic automata as briefly discussed in Remark~\ref{rem:approxModel-checking}.

We also plan to extend the work on approximated minimization~\cite{BacciBLM17,BacciBLM18} to the case of probabilistic automata and explore the possible relation between the probabilistic bisimilarity distance with  more expressive logics for concurrent probabilistic systems~\cite{ChatterjeeAMR08,ChatterjeeAMR10,Mio12}.

\paragraph*{\textbf{Acknowledgments}} The authors are grateful to the referees for their constructive feedback.

\bibliographystyle{alpha}
\bibliography{mdpHKdist}

\end{document}

